\documentclass[a4paper,twocolumn,11pt]{quantumarticle}
\pdfoutput=1
\usepackage[utf8]{inputenc}
\usepackage[english]{babel}
\usepackage[T1]{fontenc}

\usepackage{graphicx}
\usepackage{dcolumn}
\usepackage{bm}

\usepackage{etoolbox}

\usepackage{amsmath}
\usepackage{amsthm}
\usepackage{tikz}
\usepackage{extpfeil}
\usepackage{tikz-cd}
\usetikzlibrary{snakes}
\usepackage{physics}
\usepackage{caption}
\usepackage{subcaption}

\usepackage{appendix}
\usepackage[breaklinks=true]{hyperref}
\usepackage{breakcites}

\newextarrow{\xbigtoto}{{20}{20}{20}{20}}
   {\bigRelbar\bigRelbar{\bigtwoarrowsleft\rightarrow\rightarrow}}
\newextarrow{\xbigto}{{20}{20}}
   {\bigRelbar{\bigtwoarrowsleft\rightarrow}}
   
\newtheorem{theorem}{Theorem}
\newtheorem{definition}[theorem]{Definition}

\newtheorem{conjecture}[theorem]{Conjecture}

\newtheorem{example}{Example}

\begin{document}

\title{Contextuality in the Bundle Approach, n-Contextuality, and the Role of Holonomy}

\author{Sidiney B. Montanhano}
\email{s226010@dac.unicamp.br}
\affiliation{Departamento de Matemática Aplicada, Instituto de Matemática, Estatística e Computação Científica, Universidade Estadual de Campinas, 13083-859 Campinas, São Paulo, Brazil}
\orcid{0000-0002-6844-579X}

\maketitle

\begin{abstract}
  Contextuality can be understood as the impossibility to construct a globally consistent description of a model even if there is local agreement. In particular, quantum models present this property. We can describe contextuality with the bundle approach, where the scenario is represented as a simplicial complex, the fibers being the sets of outcomes, and contextuality as the non-existence of global section in the measure bundle. Using the generalization to non-finite outcome fibers, we built in details the concept of measure bundle, demonstrating the Fine-Abramsky-Brandenburger theorem for the bundle formalism. We introduce a hierarchy called n-contextuality to explore the dependence of contextual behavior of a model to the topology of the scenario, following the construction of it as a simplicial complex. With it we exemplify the dependence on higher homology groups and show that GHZ models, thus quantum theory, has all levels of the hierarchy. Also, we give an example of how non-trivial topology of the scenario result an increase of contextual behavior. For the first level of the hierarchy, we construct the concept of connection through Markov operators for the measure bundle. Taking the case of equal fibers we can identify the outcomes as the basis of a vector space, that transform according to a group extracted from the connection. We thus show that contextuality has a relationship with the non-triviality of the holonomy group in the respective frame bundle.
\end{abstract}

\section{Introduction}

\label{I}

Of all the descriptions of nature throughout human history, quantum behavior is the strangest. Not that it is necessary to satisfy our vision, but because it is after more than a century a formalism without a clear explanation that gives a reasonably intuitive panorama, which follows of our classic intuition that cannot confront quantum behavior and its with non-classical results. To make explicit such non-classicality, which appears in the early years of the development of quantum theory \cite{PhysRev.47.777}, the concept of nonlocality has emerged. With theoretical work first by Bell \cite{PhysicsPhysiqueFizika.1.195}, and with experiments that culminated in important results \cite{cite-key,PhysRevLett.115.250401,PhysRevLett.115.250402}, nonlocality finally explained the impossibility of understanding quantum theory by classical means. 

It is possible to describe the quantum theory classically, as in Bohm’s interpretation \cite{PhysRev.85.166}, but they all violate some property accepted as classical, as the speed limit for localized influence, and independence of the context in which the interaction occurs. This last property defines the concept of contextuality initially proposed in Ref. \cite{Kochen1975}, and it handles the inability to describe even with extra variables the results of a given physical system. As shown in Ref. \cite{Abramsky_2011}, nonlocality is a special case of contextuality, and one can also show from the topos formalism for quantum theory that contextuality is at the origin of all the main properties that make it what it is \cite{doring2020contextuality}, thus being at the origin of non-classicality. It can be informally explained as the inability to a global agreement even if local agreement occurs \cite{abramsky2020contextuality}, and as already showed \cite{Abramsky2015ContextualityCA}, it can be understood as the inability to remove a paradox from a data set, even including new variables, thus being described as the emergence of fundamental paradoxes \cite{Abramsky_2019}.

To understand contextual behavior, and to build tools in the study of specific models, several formalisms were constructed with their own limitations. For example, Contextuality-by-Default \cite{Dzhafarov1,Dzhafarov2,Dzhafarov3} seeks to suppose dependence in contexts from the beginning, on the level of random variables. Generalized contextuality \cite{Spekkens_2005} defines contextuality for the entire operational framework, and not just to measurements. Exclusivity graphs \cite{Cabello_2014} studies the maximum limits of violations with graph properties generated by the relationship between measurement outcomes, enabling the identification of the set of quantum correlations. The Sheaf approach to contextuality \cite{Abramsky_2011} seeks to encode contextuality through a structure in category theory, enabling equivalence between the absence of a global section and contextual behavior. It also studies the topological property of the models, and the relationship between contextuality and cohomology \cite{Car__2017,caru2018complete,montanhano2021contextuality}.

As a more recent example of formalism for the usual contextuality, we have the bundle approach, where the compatibility of measurements is codified in a hypergraph and for each vertex, understood as a measurement, we have the outcomes. The contextual behavior follows from the impossibility of globally codify the bundle constructed with these fibers, where the base is the compatibility hypergraph. Thus, contextuality appears as a property of the bundle. Its use already appeared as a graphic resource in the sheaf approach \cite{Beer_2018}, but its due formalism occurred later in Ref. \cite{Terra_2019_FibreBundle}.

In the first part of the present paper, Section \ref{II}, we make a detailed formalization for the bundle approach (\ref{IIB}), constructing the simplicial complex structure that is usually imposed in the compatibility hypergraphs (\ref{IIA}). The equality of the fibers is not supposed, putting the approach on the same level as the sheaf approach (\ref{IIC}). Another generalization is the use of non-finite fibers, as occurs in Ref. \cite{barbosa2019continuousvariable}. The version in the bundle approach of the Fine–Abramsky–Brandenburger theorem is presented (\ref{IIE}), identifying different definitions of noncontextuality, and enabling the identification of nonlocality as a special case of contextuality.

The second part, Section \ref{III}, studies the relation between contextuality and the topology of the scenario. The Vorob’ev theorem is revised, and even its property which comes from a simplification of the bundle is not a topological invariant, as an example shows, it can still suggest a geometric view of contextuality (\ref{IIIA}). By defining a hierarchy of probabilistic contextuality called $n$-contextuality (\ref{IIIB}), indexed by the maximal dimension of the contexts considered when probing the contextual behavior via contextual fraction \cite{Abramsky_2017}, we present an example of an empirical model where contextuality appears only when taking the second homological group into account ($2$-contextuality) and find an example of $n$-contextuality for all $n$ in quantum theory. We also find evidence of the dependence between the contextuality of a model and the topology of the underlying scenario (\ref{IIIC}).

Restricting to a finite fiber and graph as scenarios ($1$-contextuality) in the third part, Section \ref{IV}, the information from the previously static probabilistic outcome table is codified in Markov kernels, such that it can be understood as the connections between fibers (\ref{IVB}), flirting with the idea of dynamics in joint measurements. Singular value decomposition extracts a group element from each Markov kernel (\ref{IVC}), allowing the study of parallel transport in the bundle and the formalization of the contextual holonomy group (\ref{IVD}). We also define the respective discrete curvature in the measurement bundle with $1$-contextuality by the discrete exterior derivative using the generalized Stokes’ theorem, for cases with non-singular Markov kernels (\ref{IVF}).

Finally, we comment the results and future paths in Section \ref{V}. Examples are given throughout the paper (\ref{IID},\ref{IIF},\ref{IVE}).

\section{Bundle approach}
\label{II}

\subsection{Scenario}
\label{IIA}

Contextuality can be understood as a local concordance in the sense of compatibility of fundamental objects, usually called measurements because of physical applications, that can’t be extended to a global concordance. To formalize this idea, let’s first define how to write this concordance generically\footnote{In literature one usually consider compatibility hypergraphs \cite{amaral2017geometrical,Amaral2018}.}. 

\begin{definition}[Hypergraph]
A hypergraph is a triple $\left(X,\mathcal{C},\epsilon\right)$, where $X$ is a set of objects called vertices, $\mathcal{C}$ is a set of objects called hyperedges, and $\epsilon$ is the incidence weight $\epsilon:X\times\mathcal{C}\to\left\{0,1\right\}$ that codify if a vertex $m\in X$ is in a hyperedge $U\in\mathcal{C}$ ($\epsilon\left(m,U\right)=1$) or not ($\epsilon\left(m,U\right)=0$).
\end{definition}

This definition enables different hyperedges with the same vertices, thus encoding more information than just a relation between vertices. We might interpret it as different relations with the same objects, which can be important if one is working with relations in data analysis \cite{abramsky2020contextuality} or contextuality of ambiguous phrases in language \cite{Wang2021}. For our purposes, we can restrict to just one hyperedge by a set of vertices. Also, as we will study how the relationship between the objects interferes in the formed geometry, unrelated objects can be disregarded, restrict ourselves to the case that there is always at least one relation for any object.

\begin{definition}[Hypergraph realization as a covering]
A hypergraph $\left(X,\mathcal{C},\epsilon\right)$ is realized in the power set $\mathbb{P}\left(X\right)$ as a covering if it satisfies
\begin{itemize}
    \item $\mathcal{C}\subset\mathbb{P}\left(X\right)$;
    \item $\cup_{U_{i}\in \mathcal{C}}U_{i}=X$;
    \item $\emptyset\in \mathcal{C}$.
\end{itemize}
\end{definition}

From now on, we can omit the incidence weight.

For empirical realizations of these objects, their interpretation as a classical description, and our intuition of marginalization, the hyperedges must have the structure of a Boolean lattice, which allows the codification of them as a (abstract) simplicial complex.

\begin{definition}[Structure of a simplicial complex]
A hypergraph $\left(X,\mathcal{C}\right)$ with covering representation has the structure of a simplicial complex if it satisfies
\begin{itemize}
    \item if $U,V\in \mathcal{C}$, then $U\cap V\in \mathcal{C}$ (intersection closure); 
    \item if $U,V\in \mathcal{C}$, and $U\subset V$ then $V-U\in \mathcal{C}$ (complementation closure).
\end{itemize}
The elements of $\mathcal{C}$ will be called simplices.
\end{definition}

We can now define the initial concept to consider contextuality.

\begin{definition}[Scenario]
Let $X$ be a collection of fundamental objects called simple measurements, in the sense of being atomic. On $X$ we define a covering $\mathcal{C}\subset \mathbb{P}\left(X\right)$ such that it satisfies the structure of a simplicial complex. We call the elements of $\mathcal{C}$ contexts of the measurement set $X$, which are compatibles if they are in the same context. A pair $\mathbf{X}=\left(X,\mathcal{C}\right)$ is called a scenario.
\end{definition}

This definition implies that there are maximal contexts ($U\in \mathcal{C}$ is maximal if for all $V\in \mathcal{C}$ and $U\subset V$, then $U=V$), and minimal contexts ($U\in \mathcal{C}$ is minimal if for all $V\in \mathcal{C}$ and $V\subset U$, then $U=V$). We will denote the set of maximal contexts as $\mathcal{C}^{+}$, and the set of minimal contexts as $\mathcal{C}^{-}$. By the Boolean structure, we can limit our study to minimal contexts, and not the simple measurements. Therefore, we will call by measurements the minimal contexts and treat them as vertices of the simplicial complex\footnote{One can ask about the fundamental nature of measurements because of the coarse-graining to minimal contexts, relating it with a loss of information \cite{duarte2020investigating}. Such problem only appears here because we treat measurements as the fundamental objects. A more natural description is by the notion of coverage, and consequently of site, for the category of contexts. A site can be thought as a space without points, implying such a question is empty: there are no fundamental objects, only the covering, and its construction depends on how refined is the observer’s tools. For this, a categorical treatment of the scenario is necessary \cite{Karvonen_2019}. Such a notion of fundamental objects is allowed by the Boolean lattice, but is of a metaphysical origin, as appears in classical physics \cite{Del_Santo_2019}.}.

\subsection{Measurable bundle}

\label{IIB}

A bundle \cite{FibreBundle} is just a map $\pi :E\to B$ between a total set $E$ and a base set $B$, where the sets $F_{b}=\pi^{-1}(b)$ for each $b\in B$ are the fibers over the point $b$. One can impose topological structure to the objects of a bundle, defining a topological bundle. A fibration is a topological bundle where the fibers must be homotopy equivalent\footnote{There can be more properties imposed that defines different notions of fibrations, usually in how to treat homotopy lifting \cite{Hatcher:478079}.}. A fiber bundle is a topological bundle where the function $\pi :E\to B$ between topological spaces $E$ (the total space) and $B$ (the base space) is continuous, the fibers must be homeomorphic and it must satisfy the local triviality condition: there is a neighborhood $U\subset B$ of $b$ such that $\pi^{-1}(U)$ is homeomorphic to the Cartesian product $U\times F$.

An important concept from fiber bundles are the sections, the elements of the fibers. They can be understood as continuous functions $s_{U}:U\to E$ with $U$ the neighborhood of $b\in B$ satisfying the local triviality condition, such that $\pi\circ s_{U}$ is the identity of $U$. As we will handle with bundles, the continuity property of the sections will be excluded, keeping only the property of $s_{U}$ be the local inverse of $\pi$. If $U=B$ the section is said global, if not it is said local.

The Bundle approach is a generalization of the first introduced fibre bundle approach \cite{Terra_2019_FibreBundle} by extending to non-finite fibers and allowing they differ from each other. The main idea is to define a local triviality condition in a measure sense, where the neighborhood are the contexts, and contextuality is the failure to turn this local condition to a global one. The construction follows analogously to the fibre bundle approach, but a cautions use of objects is necessary. Unlike the usual presentation of bundles that starts from the total set $E$, the construction starts with the base and the sections, in analogy with the fiber bundle construction theorem \cite{BundleNorman}, a consequence of the incompatibility between measurements.

\begin{definition}
An event bundle $(X,\mathcal{C},\{s_{U}\}_{U\in\mathcal{C}^{+}})$ is a scenario $(X,\mathcal{C})$ with sections $s_{U}(V)=(V,O^{V})$, where $V\in \mathcal{C}$ is a sub-context of $U\in\mathcal{C}^{+}$ and $O^{V}$ is the set of outcomes over the context $V$, which satisfies
\begin{itemize}
    \item $O^{U}=O^{V}\times O^{U-V}$, thus defining the marginalization $s_{U}(V)=s_{U}(U)|_{V}$;
    \item there is a global section $s_{X}$ that marginalizes to the contexts in $\mathcal{C}$, i.e. $s_X(V)=s_U(V)$ for all $U\in \mathcal{C}^+$ and $V\subset U$, $V\in\mathcal{C}$.
\end{itemize}
\end{definition}

The sections codify the local triviality condition of an event bundle and one can see that it must be globally trivial. The idea is to turn the elements of each fiber, called events, as the marginalization of well-defined global events.

We will first consider generic fibers, not only finite ones, thus a $\sigma$-algebra of each set of outcomes must be made explicit, making them measurable spaces. We will denote each fiber of a context $U$ as $\textbf{O}^{U}=(O^{U},\Sigma^{U})$.

\begin{definition}
A measurable bundle $(X,\mathcal{C},\{s_{U}\}_{U\in\mathcal{C}^{+}})$ is an event bundle with sections defining the fibers as measurable spaces, i.e. $s_{U}(V)=(V,\mathbf{O}^{V})$ satisfying
\begin{itemize}
    \item $\mathbf{O}^{U}=\mathbf{O}^{V}\times \mathbf{O}^{U-V}$, thus defining the marginalization $s_{U}(V)=s_{U}(U)|_{V}$;
    \item there is a global section $s_{X}$ that marginalizes to the contexts in $\mathcal{C}$, i.e. $s_X(V)=s_U(V)$ for all $U\in \mathcal{C}^+$ and $V\subset U$, $V\in\mathcal{C}$.
\end{itemize}
\end{definition}

Effectively, the events over each context are the elements of the $\sigma$-algebra, and we will call them events in the following, keeping the $\sigma$-algebra implicit when they are trivial\footnote{Usually, the events are the elements of $O^{U}$. But this only holds because one is using a $\sigma$-algebra with atomic objects, usually the discrete one.}. Measurement bundle is the generalization in the bundle approach to the notion of measurement scenario, a scenario with outcomes attached in a consistent way to the measurements.

One can request a topology of a measurable bundle such that it could be seen as a fiber bundle when the fibers are homeomorphic. Even if it is initially unnecessary once one is interested in the behavior of measures on the measurable bundle, there are open questions about the relation of such behaviors and topological data. There are two natural topologies for a scenario, the geometric representation of the simplicial complex and the discrete topology with minimal contexts as points. A fiber $\textbf{O}^{U}$ can be equipped with the thickest topology that contains its $\sigma$-algebra,
\begin{widetext}
\begin{equation}
    \tau(\Sigma^{U})=\bigcap\{\tau\subseteq \mathbb{P}(O^{U})|\tau\text{ is a topology on $O^{U}$ with $\Sigma^{U}\subseteq\tau$}\},
\end{equation}
\end{widetext}
thus respecting the measurable bundle structure.

\subsection{Measure bundles}
\label{IIC}

Given a measurable bundle, one can ask for measures defined on it. As we are interested in probabilistic measures and in analogy with axiomatic quantum theory, one can define a state as a function that gives for each context a probabilistic measure.

\begin{definition}
A state is a map $P::U\in\mathcal{C}\to \mu^{U}$ from a context to a probabilistic measure on its events, $\mu^{U}:\Sigma^{U}\to [0,1]$, such that $\int_{O^{U}}\mu^{U}=1$, where the integral is defined on the $\sigma$-algebra of the fiber.
\end{definition}

From now on, when we write measure, it will be understood that it is a probabilistic measure, leaving the exploration of other kind of measures, like signed measures, to a future work. The notion of state codify the new data necessary to define the main object of any approach to contextuality.

\begin{definition}[Empirical model]
An empirical model is defined as an measurable bundle equipped with a fixed state $P$ acting on it.
\end{definition}

One can promote a section $s_{U}$ of an measurable bundle to a section which also takes a context $U$ and gives a measure $\mu^{U}$ on the fiber $\mathbf{O}^{U}$, thus $s_{U}(V)=(V,\mathbf{O}^{V},\mu^{U}|_{V})$. But one cannot readily define a bundle once here there is no imposition of a global triviality as for event bundles, it needs measurable transition functions between the measure spaces at the intersections between contexts, which were take implicitly for event bundles due to global triviality.

\begin{definition}
A measure bundle $(X,\mathcal{C},\{s_{U}\}_{U\in\mathcal{C}^{+}},\{t_{U,V}\}_{U,V\in \mathcal{C},U\cap V\neq\emptyset})$ is an empirical model with a collection of measurable transition functions $t_{U,V}:s_{U}|_{U\cap V}\to s_{V}|_{U\cap V}$.
\end{definition}

A natural choice of transition functions is the trivial one, know in literature as no-disturbance condition on empirical models.

\begin{definition}[No-disturbance]
An empirical model is called non-disturbing if for any two contexts $U_{1}$ and $U_{2}$ with $U=U_{1}\cap U_{2}\neq\emptyset$, then $\mu^{U_{1}}|_{U}=\mu^{U_{2}}|_{U}=\mu^{U}$ holds. This is equivalent to the transition functions $\{t_{U,V}\}$ being trivial, i.e. they are the identity of the intersection’s fiber in the bundle representation. Otherwise the empirical model is called disturbing.
\end{definition}

Here the terms measure bundle and empirical model will be used interchangeably, and the no-disturbance condition will be supposed.

\subsection{First examples of measure bundles}
\label{IID}

\begin{example}[Trivial] 
The standard trivial example of a non-disturbing measure bundle, analogous to the informal description of a trivial bundle, is given by the table of section probabilities of Table \ref{Trivial bundle}.
\begin{figure}
\centering
\scalebox{0.3}{
\begin{tikzpicture}

\draw [ultra thick] (-9,12) -- (-3,10) node[right] {\Huge{$\ a$}};
\draw [ultra thick] (-3,10) -- (0,13) node[right] {\Huge{$\ b$}}; 
\draw [ultra thick] (0,13) -- (-2,15.5) node[right] {\Huge{$\ c$}}; 
\draw [ultra thick] (-2,15.5) -- (-6,15)
node[left] {\Huge{$d\ $}}; 
\draw [ultra thick] (-6,15) -- (-9,12) node[left] {\Huge{$e\ $}};

\filldraw [black] (-9,12) circle (4pt);
\filldraw [black] (-3,10) circle (4pt);
\filldraw [black] (0,13) circle (4pt);
\filldraw [black] (-2,15.5) circle (4pt);
\filldraw [black] (-6,15) circle (4pt);

\draw[loosely dotted, ultra thick] (-9,12) -- (-9,21);
\draw[loosely dotted, ultra thick] (-3,10) -- (-3,19);
\draw[loosely dotted, ultra thick] (0,13) -- (0,22);
\draw[loosely dotted, ultra thick] (-2,15.5) -- (-2,24.5);
\draw[loosely dotted, ultra thick] (-6,15) -- (-6,24);

\draw [ultra thick] (-9,19) -- (-3,17);
\draw [ultra thick] (-3,17) -- (0,20) node[right] {\Huge{$\ 1$}}; 
\draw [ultra thick] (0,20) -- (-2,22.5); 
\draw [ultra thick] (-2,22.5) -- (-6,22); 
\draw [ultra thick] (-6,22) -- (-9,19);

\filldraw [black] (-9,19) circle (4pt);
\filldraw [black] (-3,17) circle (4pt);
\filldraw [black] (0,20) circle (4pt);
\filldraw [black] (-2,22.5) circle (4pt);
\filldraw [black] (-6,22) circle (4pt);

\draw [ultra thick] (-9,21) -- (-3,19);
\draw [ultra thick] (-3,19) -- (0,22) node[right] {\Huge{$\ 0$}}; 
\draw [ultra thick] (0,22) -- (-2,24.5); 
\draw [ultra thick] (-2,24.5) -- (-6,24); 
\draw [ultra thick] (-6,24) -- (-9,21);

\filldraw [black] (-9,21) circle (4pt);
\filldraw [black] (-3,19) circle (4pt);
\filldraw [black] (0,22) circle (4pt);
\filldraw [black] (-2,24.5) circle (4pt);
\filldraw [black] (-6,24) circle (4pt);

\end{tikzpicture}}
\caption{The visualization of the trivial example of a measure bundle through its bundle diagram. The global sections are explicit, and determined by the fiber elements.}
\label{Trivial bundle Fig}
\end{figure}

\begin{table}
  \centering
\begin{tabular}{|c|c|c|c|c|}
\hline\noalign{\smallskip}
& \large{$00$} & \large{$01$} & \large{$10$} & \large{$11$} \\
\noalign{\smallskip}\hline\noalign{\smallskip}
\large{$ab$} & \large{$\frac{1}{2}$} & \large{$0$} & \large{$0$} & \large{$\frac{1}{2}$} \\
\noalign{\smallskip}\hline\noalign{\smallskip}
\large{$bc$} & \large{$\frac{1}{2}$} & \large{$0$} & \large{$0$} & \large{$\frac{1}{2}$} \\
\noalign{\smallskip}\hline\noalign{\smallskip}
\large{$cd$} & \large{$\frac{1}{2}$} & \large{$0$} & \large{$0$} & \large{$\frac{1}{2}$} \\
\noalign{\smallskip}\hline\noalign{\smallskip}
\large{$de$} & \large{$\frac{1}{2}$} & \large{$0$} & \large{$0$} & \large{$\frac{1}{2}$} \\
\noalign{\smallskip}\hline\noalign{\smallskip}
\large{$ea$} & \large{$\frac{1}{2}$} & \large{$0$} & \large{$0$} & \large{$\frac{1}{2}$} \\
\noalign{\smallskip}\hline
\end{tabular}
\caption{Table of outcome probabilities of each context of the trivial example.}
    \label{Trivial bundle}
\end{table}
Here, and for all the explicit examples of this paper, we have finite simplicial complex as the base and finite fibers. The contexts of this example are
\begin{equation}
    \mathcal{C}=\left\{ab,bc,cd,de,ea,a,b,c,d,e\right\},
\end{equation} 
with measurements 
\begin{equation}
\mathcal{C}^{-}=\left\{a,b,c,d,e\right\}
\end{equation} 
and maximal contexts
\begin{equation}
\mathcal{C}^{+}=\left\{ab,bc,cd,de,ea\right\}.
\end{equation} 
All the fibers for measurements are the same, $O^{i}=\left\{0,1\right\}$. Its table gives the probabilities for each event of each maximal context. Examples with finite fibers can be described by the values of each element of $O^{U}$, using the discrete $\sigma$-algebra. The bundle diagram is the possibilistic coarse-graining of the model, with only the non-null events appearing. All the events of all the intersections have probability $\frac{1}{2}$. It can be described as two global events with probability $\frac{1}{2}$ each. Its no-disturbance follows from the fact that all the marginalization over the probabilities of measurement outcomes also yields $\frac{1}{2}$ each.
\end{example}

\begin{example}[Liar cycle]
An example of non-disturbing measure bundle with no global section is the liar cycle, Table \ref{Liar cycle}, here with five measurements. It could be understood as a set of individuals saying that the next one will say the truth, cyclically, but the last one saying that the first one lied, such that a paradox occurs. The structure of the event bundle is the same as the previous example, it only differs in the measure defined on $\textbf{O}^{ea}$. The bundle is non-disturbing, but global sections that define the bundle as a measure on global events do not exist.
\begin{figure}
  \centering
\scalebox{0.3}{
\begin{tikzpicture}

\draw [ultra thick] (-9,12) -- (-3,10) node[right] {\Huge{$\ a$}};
\draw [ultra thick] (-3,10) -- (0,13) node[right] {\Huge{$\ b$}}; 
\draw [ultra thick] (0,13) -- (-2,15.5) node[right] {\Huge{$\ c$}}; 
\draw [ultra thick] (-2,15.5) -- (-6,15)
node[left] {\Huge{$d\ $}}; 
\draw [ultra thick] (-6,15) -- (-9,12) node[left] {\Huge{$e\ $}};

\filldraw [black] (-9,12) circle (4pt);
\filldraw [black] (-3,10) circle (4pt);
\filldraw [black] (0,13) circle (4pt);
\filldraw [black] (-2,15.5) circle (4pt);
\filldraw [black] (-6,15) circle (4pt);

\draw[loosely dotted, ultra thick] (-9,12) -- (-9,21);
\draw[loosely dotted, ultra thick] (-3,10) -- (-3,19);
\draw[loosely dotted, ultra thick] (0,13) -- (0,22);
\draw[loosely dotted, ultra thick] (-2,15.5) -- (-2,24.5);
\draw[loosely dotted, ultra thick] (-6,15) -- (-6,24);

\filldraw [black] (-9,21) circle (4pt);
\filldraw [black] (-3,19) circle (4pt);
\filldraw [black] (0,22) circle (4pt);
\filldraw [black] (-2,24.5) circle (4pt);
\filldraw [black] (-6,24) circle (4pt);

\filldraw [black] (-9,19) circle (4pt);
\filldraw [black] (-3,17) circle (4pt);
\filldraw [black] (0,20) circle (4pt);
\filldraw [black] (-2,22.5) circle (4pt);
\filldraw [black] (-6,22) circle (4pt);

\draw [ultra thick] (-9,19) -- (-3,19);
\draw [ultra thick] (-3,17) -- (0,20) node[right] {\Huge{$\ 1$}}; 
\draw [ultra thick] (0,20) -- (-2,22.5); 
\draw [ultra thick] (-2,22.5) -- (-6,22); 
\draw [ultra thick] (-6,22) -- (-9,19);

\draw [ultra thick] (-9,21) -- (-3,17);
\draw [ultra thick] (-3,19) -- (0,22) node[right] {\Huge{$\ 0$}}; 
\draw [ultra thick] (0,22) -- (-2,24.5); 
\draw [ultra thick] (-2,24.5) -- (-6,24); 
\draw [ultra thick] (-6,24) -- (-9,21);

\end{tikzpicture}}
\caption{The visualization of the liar cycle example through its bundle diagram. There isn't any global event, caused by the swap of elements in the path $ea$.}
\label{Liar cycle  Fig}
\end{figure}

\begin{table}
  \centering
\begin{tabular}{|c|c|c|c|c|}
\hline\noalign{\smallskip}
& \large{$00$} & \large{$01$} & \large{$10$} & \large{$11$} \\
\noalign{\smallskip}\hline\noalign{\smallskip}
\large{$ab$} & \large{$\frac{1}{2}$} & \large{$0$} & \large{$0$} & \large{$\frac{1}{2}$} \\
\noalign{\smallskip}\hline\noalign{\smallskip}
\large{$bc$} & \large{$\frac{1}{2}$} & \large{$0$} & \large{$0$} & \large{$\frac{1}{2}$} \\
\noalign{\smallskip}\hline\noalign{\smallskip}
\large{$cd$} & \large{$\frac{1}{2}$} & \large{$0$} & \large{$0$} & \large{$\frac{1}{2}$} \\
\noalign{\smallskip}\hline\noalign{\smallskip}
\large{$de$} & \large{$\frac{1}{2}$} & \large{$0$} & \large{$0$} & \large{$\frac{1}{2}$} \\
\noalign{\smallskip}\hline\noalign{\smallskip}
\large{$ea$} & \large{$0$} & \large{$\frac{1}{2}$} & \large{$\frac{1}{2}$} & \large{$0$} \\
\noalign{\smallskip}\hline
\end{tabular}
\caption{Table of outcome probabilities of each context of the liar cycle example.}
\end{table}
The liar $n$-cycle is an important example for $n$-cycle scenarios, i.e. scenarios with $n$ measurements organized as a cycle, once one can construct any paradoxical behavior on $n$-cycle scenarios by them \cite{PhysRevA.104.022201}. It is also important as an example of nonlocal “superquantum” correlations genarated by the well-known Popescu-Rohrlich boxes \cite{Popescu1994}.
\label{Liar cycle}
\end{example}

\begin{example}[KCBS]
The next example looks like an extreme version of the liar cycle, but in fact just repeat the change of elements of the outcome fiber in every context, or in the logical representation, everyone are saying the next one is a liar, Table \ref{KCBS}. Again the event bundle is the same, the difference is the measures, which is the bundle description of the Klyachko-Can-Binicioglu-Shumovsky pentagram scenario \cite{PhysRevLett.101.020403} realizable in quantum theory. There isn't a global section, which can be view by the bundle diagram.
\begin{figure}
  \centering
\scalebox{0.3}{
\begin{tikzpicture}

\draw [ultra thick] (-9,12) -- (-3,10) node[right] {\Huge{$\ a$}};
\draw [ultra thick] (-3,10) -- (0,13) node[right] {\Huge{$\ b$}}; 
\draw [ultra thick] (0,13) -- (-2,15.5) node[right] {\Huge{$\ c$}}; 
\draw [ultra thick] (-2,15.5) -- (-6,15)
node[left] {\Huge{$d\ $}}; 
\draw [ultra thick] (-6,15) -- (-9,12) node[left] {\Huge{$e\ $}};

\filldraw [black] (-9,12) circle (4pt);
\filldraw [black] (-3,10) circle (4pt);
\filldraw [black] (0,13) circle (4pt);
\filldraw [black] (-2,15.5) circle (4pt);
\filldraw [black] (-6,15) circle (4pt);

\draw[loosely dotted, ultra thick] (-9,12) -- (-9,21);
\draw[loosely dotted, ultra thick] (-3,10) -- (-3,19);
\draw[loosely dotted, ultra thick] (0,13) -- (0,22);
\draw[loosely dotted, ultra thick] (-2,15.5) -- (-2,24.5);
\draw[loosely dotted, ultra thick] (-6,15) -- (-6,24);

\filldraw [black] (-9,21) circle (4pt);
\filldraw [black] (-3,19) circle (4pt);
\filldraw [black] (0,22) circle (4pt);
\filldraw [black] (-2,24.5) circle (4pt);
\filldraw [black] (-6,24) circle (4pt);

\filldraw [black] (-9,19) circle (4pt);
\filldraw [black] (-3,17) circle (4pt);
\filldraw [black] (0,20) circle (4pt);
\filldraw [black] (-2,22.5) circle (4pt);
\filldraw [black] (-6,22) circle (4pt);

\draw [ultra thick] (-9,19) -- (-3,19);
\draw [ultra thick] (-3,19) -- (0,20) node[right] {\Huge{$\ 1$}}; 
\draw [ultra thick] (0,20) -- (-2,24.5); 
\draw [ultra thick] (-2,24.5) -- (-6,22); 
\draw [ultra thick] (-6,22) -- (-9,21);

\draw [ultra thick] (-9,21) -- (-3,17);
\draw [ultra thick] (-3,17) -- (0,22) node[right] {\Huge{$\ 0$}}; 
\draw [ultra thick] (0,22) -- (-2,22.5); 
\draw [ultra thick] (-2,22.5) -- (-6,24); 
\draw [ultra thick] (-6,24) -- (-9,19);

\end{tikzpicture}}
\caption{The visualization of the KCBS example through its bundle diagram. Every local section is a swap, implying that for an odd number of sides there isn't any well defined global event.}
    \label{KCBS Fig}
\end{figure}

\begin{table}
  \centering
\begin{tabular}{|c|c|c|c|c|}
\hline\noalign{\smallskip}
& \large{$00$} & \large{$01$} & \large{$10$} & \large{$11$} \\
\noalign{\smallskip}\hline\noalign{\smallskip}
\large{$ab$} & \large{$0$} & \large{$\frac{1}{2}$} & \large{$\frac{1}{2}$} & \large{$0$} \\
\noalign{\smallskip}\hline\noalign{\smallskip}
\large{$bc$} & \large{$0$} & \large{$\frac{1}{2}$} & \large{$\frac{1}{2}$} & \large{$0$} \\
\noalign{\smallskip}\hline\noalign{\smallskip}
\large{$cd$} & \large{$0$} & \large{$\frac{1}{2}$} & \large{$\frac{1}{2}$} & \large{$0$} \\
\noalign{\smallskip}\hline\noalign{\smallskip}
\large{$de$} & \large{$0$} & \large{$\frac{1}{2}$} & \large{$\frac{1}{2}$} & \large{$0$} \\
\noalign{\smallskip}\hline\noalign{\smallskip}
\large{$ea$} & \large{$0$} & \large{$\frac{1}{2}$} & \large{$\frac{1}{2}$} & \large{$0$} \\
\noalign{\smallskip}\hline
\end{tabular}
    \caption{Table of outcome probabilities of each context of the KCBS example.}
    \label{KCBS}
\end{table}
\end{example}

As one can see, at the level of the event bundle all the previous examples are topologically equivalent for both the discrete and the geometrical representation topologies. The first erases any data about the contexts, while the second captures the cyclic nature of the base. The topology of the fibers must be the trivial one, such that one can represent the previous examples as fiber bundles and the nonexistence of a description with global events a natural consequence of the non-trivial topology of the fiber bundle. At this point the reader must be alerted that this reasoning only works because of the triviality of these examples, where the global events are distinguishable and there is a clear topological translation. Next examples will lose such direct topological interpretation.

\subsection{Noncontextual behaviour}

\label{IIE}

There are different ways to formalize contextuality, all as the violation of some noncontextual notion previously defined. A natural one follows from seeking to include hidden variables to explain the empirical model, which is like including auxiliary variables to correct data inconsistencies and to resolve non-fundamental paradoxes \cite{abramsky2020contextuality}.

\begin{definition}[Noncontextuality by model]
An empirical model is noncontextual by model if the measure $\mu^{U}$ depends on each of the measurements that make up $U$ independently, except for hidden variables that must be statistically considered, i.e.
\begin{equation}
\mu^{U}=\int_{\Lambda}\left(p\left(\lambda\right)\prod_{V_{i}\in \mathcal{C}^{-}|_{U}}\mu^{\left(V_{i},\lambda\right)}\right)d\lambda
\end{equation}
where it is defined as the average of hidden variables in space $\Lambda$ with weight $p$ satisfying $\int_{\Lambda}p(\lambda)=1$. The measure on $\textbf{O}^{U}$ couples the measures $\mu^{V_{i}}$ of the fibers $\textbf{O}^{V_{i}}$, with $\bigcup_{i}V_{i}=U$ the decomposition of $U$ in its measurements, in the canonical way.
\end{definition}

This noncontextual model definition assumes two properties of the allowed hidden-variable model, as punctuated in Ref. \cite{barbosa2019continuousvariable}. The first is $\lambda$-independence, when the measure is context-independent, and the second is parameter-independence, when no-disturbance of the hidden variable model holds. Both and more properties of hidden variable models are studied in Ref. \cite{Brandenburger_2008}\footnote{What this definition is saying is that a noncontextual model can be understood as a coarse-graining of a measure bundle with contexts $(U,\lambda)$ such that the coupling of all the measures is globally trivial.}.

Another way to define noncontextuality is by defining an unique function that explains all the empirical model. One can conceive such a definition as asking a physical system a question through a measurement without interfering with the outcome of the result, as a reality that just needs to be discovered, but because of the contexts one cannot see it directly.

\begin{definition}[Noncontextuality by section]
An empirical model is said noncontextual by section if every local section can be extended to a global section.
\end{definition}

The notion of contextuality by section originates from the sheaf approach to contextuality \cite{Abramsky_2011}, but it was first written for bundles in the approach presented in Ref. \cite{Terra_2019_FibreBundle}.

Finally, one can start with a global measure that can be marginalized to get the measure of each context, as it is supposed in a classical model. Inspired by this there is the definition of noncontextuality as the existence of such a global measure.

\begin{definition}[Noncontextuality by marginals]
An empirical model is noncontextual by marginals if there is a measure $\mu^{X}$, defined on all the measurements set, such that we can write any measure as $\mu^{U}=\mu^{X}|_{U}$ for all $U\in\mathcal{C}$.
\end{definition}

The next theorem, well known in the sheaf approach to contextuality \cite{Abramsky_2011}, is a generalization of the same result presented in Ref. \cite{Terra_2019_FibreBundle} in the fiber bundle approach. Here we generalized the fibers to any measurable space, not just finite ones. The demonstration is inspired by the same theorem presented in Ref. \cite{barbosa2019continuousvariable}.

\begin{theorem}
[Fine-Abramsky-Brandenburger] Given an empirical model on a scenario $\textbf{X}$, the following statements are equivalent:
\begin{enumerate}
    \item the model is noncontextual by model;
    \item the model is noncontextual by marginals;
    \item the model is noncontextual by section.
\end{enumerate}
\label{FAB}
\end{theorem}

\begin{proof}
($1\to 2$) If we have an empirical model that is noncontextual by model, then for any context $U\in\mathcal{C}$ holds
\begin{align}
\begin{split}
\mu^{U} &=\int_{\Lambda}\left(p\left(\lambda\right)\prod_{V_{i}\in \mathcal{C}^{-}|_{U}}\mu^{\left(V_{i},\lambda\right)}\right)d\lambda\\
&=\int_{\Lambda}dp\left(\lambda\right)\prod_{V_{i}\in \mathcal{C}^{-}|_{U}}\mu^{\left(V_{i},\lambda\right)}.
\end{split}
\end{align}

By the same reason, we can write
\begin{equation}
\mu^{X}=\int_{\Lambda}dp\left(\lambda\right)\prod_{V_{i}\in \mathcal{C}^{-}}\mu^{\left(V_{i},\lambda\right)}
\end{equation}
defining a measure on $X$. The marginals follow from
\begin{align}
\begin{split}
\mu^{X}|_{U}&=\int_{\Lambda}dp\left(\lambda\right)\prod_{V_{i}\in\mathcal{C}^{-}}\mu^{\left(V_{i},\lambda\right)}|_{U}\\
&=\int_{\Lambda}dp\left(\lambda\right)\prod_{V_{i}\in\mathcal{C}^{-}|_{U}}\mu^{\left(V_{i},\lambda\right)}\\
&=\mu^{U},
\end{split}
\end{align}
therefore, it is a model noncontextual by marginals.

($2\to 3$) Given a measure $\mu^{X}$ of $X$, and let $s_{U}$ be a local section, $\pi\left(s_{U}\right)\left(V\right)=V$ for all $V\subset U$, $V\in\mathcal{C}$. As 
\begin{equation}
s_{U}\left(V\right)=\left(V,\textbf{O}^{V},\mu^{U}|_{V}\right)
\end{equation}
we can extend the section by
\begin{equation}
s_{X}\left(V\right)=\left(V,\textbf{O}^{V},\mu^{X}|_{V}\right)
\end{equation}
where obviously holds for all $U'\in\mathbb{P}\left(\mathcal{C}^{-}\right)$
\begin{equation}
s_{X}\left(U'\right)=\left(U',\textbf{O}^{U'},\mu^{X}|_{U'}\right).
\end{equation}
This is especially true for $U'=X$, and therefore the model is noncontextual by section.

($3\to 1$) If we have an empirical model that is noncontextual by section, then for any $U\in\mathbb{P}(\mathcal{C}^{-})$ holds
\begin{equation}
s_{X}\left(U\right)=\left(U,\textbf{O}^{U},\mu^{X}|_{U}\right).
\end{equation}
We always can write a measure on $U$ as
\begin{equation}
\mu^{U}(\sigma_{U})=\int_{\Lambda}dp(\lambda)K_{U}(\lambda,\sigma_{U})
\end{equation}
with any $\sigma_{U}\in\Sigma^{U}$, and $K$ the kernel for $U$. Choosing $\Lambda=O^{X}$, $p=\mu^{X}$, and defining $K_{U}(\lambda,\sigma_{U})=\delta_{\lambda|_{U}}(\sigma_{U})$ and the restriction function $f\left(\lambda\right)=\lambda|_{U}\in O^{U}$, we have
\begin{align}
\begin{split}
    \int_{\Lambda}dp(\lambda)K_{U}(\lambda,\sigma_{U}) 
    &=\int_{O^{X}}d\mu^{X}(\lambda)K_{U}(\lambda,\sigma_{U})\\
    &=\int_{O^{X}}d\mu^{X}(\lambda)\delta_{f\left(\lambda\right)}(\sigma_{U})\\
    &=\int_{O^{X}}d\mu^{X}(\lambda)\chi_{\sigma_{U}}\circ f(\lambda)\\
    &=\int_{O^{U}}d\left(f^{*}\mu^{X}\right)(\lambda)\chi_{\sigma_{U}}(x)\\
    &=\int_{O^{U}}d\mu^{X}(\lambda)|_{U}\chi_{\sigma_{U}}(x)\\
    &=\mu^{X}|_{U}(\sigma_{U}).
\end{split}
\end{align}
with $x\in O^{U}$ and where $\chi_{\sigma_{U}}(\lambda)=\delta_{\lambda|_{U}}(\sigma_{U})$ is the characteristic function. This tells us that the model we got is deterministic. It follows it is factorizable, since the Dirac measure $\delta_{\lambda|_{U}}$ factorizes as a product of Dirac measures,
\begin{equation}
\delta_{\lambda|_{U}}=\prod_{V_{i}\in \mathcal{C}^{-}|_{U}}\delta_{\lambda|_{U}|_{V_{i}}}=\prod_{V_{i}\in \mathcal{C}^{-}|_{U}}\delta_{\lambda|_{V_{i}}}.
\end{equation}
Being $\sigma_{U}=\prod_{{V_{i}\in \mathcal{C}^{-}|_{U}}}\sigma_{V_{i}}$, we can write
\begin{align}
\begin{split}
K_{U}\left(\lambda,\sigma_{U}\right)&=\delta_{\lambda|_{U}}(\sigma_{U})\\
&=\prod_{V_{i}\in \mathcal{C}^{-}|_{U}}\delta_{\lambda|_{V_{i}}}\left(\sigma_{V_{i}}\right)\\
&=\prod_{V_{i}\in\mathcal{C}^{-}|_{U}}K_{V_{i}}\left(\lambda,\sigma_{V_{i}}\right)
\end{split}
\end{align}
therefore, we can finally define $K_{V_{i}}\left(\lambda,\sigma_{V_{i}}\right)=\mu^{\left(V_{i},\lambda\right)}\left(\sigma_{V_{i}}\right)$ and as it also holds for $X$, we conclude the model is noncontextual by model.
\end{proof}

During the previous demonstration, we showed that a model is noncontextual if and only if it can be factorizable. This property when dealing with probabilistic measures is equivalent to Bell’s locality \cite{PhysicsPhysiqueFizika.1.195}, as explained in Ref. \cite{Abramsky_2011}. The theorem above thus relates Bell’s nonlocality as a special case of the contextuality of the respective empirical model. For our case with non-finite fibers, we generalize this identification for cases of a continuous spectrum. Another important point is the role of outcome-determinism by choosing $K_{U}(\lambda,\sigma_{U})=\delta_{\lambda|_{U}}(\sigma_{U})$. This imposition restricts the kernels one can use, restricting the previous theorem to ideal measurements \cite{Spekkens_2014}.

\subsection{Examples of contextuality in measure bundles}
\label{IIF}

The examples that will be presented here have finite outcome fibers and are non-disturbing. To check the contextual behavior, we will use the contextual fraction.

The contextual fraction is a measure of contextuality, based on the fact that a noncontextual model can be written as a convex combination of global events. The finite version was present in Ref. \cite{Abramsky_2017}, and the more general case of non-finite fibers in Ref. \cite{barbosa2019continuousvariable}.

Formally, each measurable bundle has an incidence matrix \begin{equation}
    M\left(\sigma_{U},\sigma_{\mathcal{B}}\right)=\begin{cases}
      1\text{ if } \sigma_{\mathcal{B}}|_{U}=\sigma_{U};\\
      0\text{ otherwise}.
    \end{cases}  
\end{equation}
It has the possible global events indexing the columns and the possible local events of maximal contexts indexing the rows, such that the entry will only be non-null if the local event is a restriction to the context of the global event. A non-disturbing and deterministic model to be noncontextual is equivalent to it be described as a convex combination of its global events, and therefore we will have a weight $b_{i}$ for each of them, which should add $1$. It is equivalent to
\begin{equation}
    M\vec{b}:=\vec{p}
\end{equation}
where $\vec{p}$ is the vector of the probabilities of the outcome in each context (certain care must be taken so that the vector, which originates in the usual probability table, is in the correct position about the entry in the incidence matrix).

The noncontextual fraction $NCF$ is the maximum value of $\sum_{i}b_{i}$ such that $b_{i}\geq 0$ and $\sum_{j}M_{ij}b_{j}\leq p_{i}$, i.e., 
\begin{equation}
    NCF:=\max_{\vec{b}}\left\{\sum_{i}b_{i};b_{i}\geq 0,\sum_{j}M_{ij}b_{j}\leq p_{i}\right\}.
\end{equation}
The contextual fraction is then defined as $CF=1-NCF$.\footnote{The contextual fraction is related to the contextuality inequalities in the literature and has the necessary properties to be considered a "good" measure of contextuality. In case of interest in such properties, the references already mentioned have more details, and Ref. \cite{amaral2017geometrical} has a good review.}

\begin{example}[Trivial] 
The example of Table \ref{Trivial bundle} can be describe as the combination of two global events,
\begin{equation}
    s_{X}^{trivial}=\left\{\frac{1}{2}(abcde\to 00000)+\frac{1}{2}(abcde\to 11111)\right\}
\end{equation}
defining a global section of the probabilistic bundle. By Theorem \ref{FAB}, the example is noncontextual, in agreement with noncontextual fraction $NCF=1$.
\end{example}

\begin{example}[Liar cycle]
The example of Table \ref{Liar cycle} can’t be described as a combination of global events, thus there isn’t a global section that describes it. The cause is the swap in the context $ea$, giving a contextual behaviour for this probabilistic bundle, with noncontextual fraction $NCF=0$. A liar cycle is an example of strong contextuality \cite{Abramsky_2011}, once there is not any well-defined global event.
\end{example}

\begin{example}[KCBS]
The example of Table \ref{KCBS} follows the same reasoning as the previous one. A combination of global events can’t describe it, because of an odd number of swaps. The noncontextual fraction is $NCF=0$, demonstrating the contextual behavior of this example.
\end{example}

\begin{example}[Hardy]
The example of Table \ref{Hardy} is more complicated, but one still can directly see if it is contextual or not just using the events. It was first introduced in Ref. \cite{Cabello_2013} as a Hardy-like model \cite{PhysRevLett.68.2981,PhysRevLett.71.1665} with contextual behavior in quantum theory. The model is non-disturbing, with noncontextual fraction $NCF=\frac{7}{9}$, thus being contextual. It has global events, starting in $(ab\to 00)$, $(cd\to 00)$ and $(ea\to 00)$, but there are events that can’t be extended to global ones, $(ab\to 01)$, $(ab\to 10)$, $(cd\to 01)$, $(cd\to 10)$, $(ea\to 01)$ and $(ea\to 10)$. A model that presents global events but also presents non-extendable local events is called possibilistic contextual \cite{Abramsky_2011}, as this example.

\begin{figure}
  \centering
\scalebox{0.3}{
\begin{tikzpicture}
\draw [ultra thick] (-9,12) -- (-3,10) node[right] {\Huge{$\ a$}};
\draw [ultra thick] (-3,10) -- (0,13) node[right] {\Huge{$\ b$}}; 
\draw [ultra thick] (0,13) -- (-2,15.5) node[right] {\Huge{$\ c$}}; 
\draw [ultra thick] (-2,15.5) -- (-6,15)
node[left] {\Huge{$d\ $}}; 
\draw [ultra thick] (-6,15) -- (-9,12) node[left] {\Huge{$e\ $}};
\filldraw [black] (-9,12) circle (4pt);
\filldraw [black] (-3,10) circle (4pt);
\filldraw [black] (0,13) circle (4pt);
\filldraw [black] (-2,15.5) circle (4pt);
\filldraw [black] (-6,15) circle (4pt);
\draw[loosely dotted, ultra thick] (-9,12) -- (-9,21);
\draw[loosely dotted, ultra thick] (-3,10) -- (-3,19);
\draw[loosely dotted, ultra thick] (0,13) -- (0,22);
\draw[loosely dotted, ultra thick] (-2,15.5) -- (-2,24.5);
\draw[loosely dotted, ultra thick] (-6,15) -- (-6,24);
\filldraw [black] (-9,21) circle (4pt);
\filldraw [black] (-3,19) circle (4pt);
\filldraw [black] (0,22) circle (4pt);
\filldraw [black] (-2,24.5) circle (4pt);
\filldraw [black] (-6,24) circle (4pt);
\filldraw [black] (-9,19) circle (4pt);
\filldraw [black] (-3,17) circle (4pt);
\filldraw [black] (0,20) circle (4pt);
\filldraw [black] (-2,22.5) circle (4pt);
\filldraw [black] (-6,22) circle (4pt);
\draw [ultra thick] (-9,19) -- (-3,19);
\draw [ultra thick] (-3,19) -- (0,20) node[right] {\Huge{$\ 1$}}; 
\draw [ultra thick] (0,20) -- (-2,24.5);
\draw [ultra thick] (-2,24.5) -- (-6,22); 
\draw [ultra thick] (-6,22) -- (-9,21);
\draw [ultra thick] (-9,21) -- (-3,17);
\draw [ultra thick] (-3,17) -- (0,22) node[right] {\Huge{$\ 0$}}; 
\draw [ultra thick] (0,22) -- (-2,22.5);
\draw [ultra thick] (-2,22.5) -- (-6,24); 
\draw [ultra thick] (-6,24) -- (-9,19);
\draw [ultra thick] (-9,21) -- (-3,19);
\draw [ultra thick] (-3,19) -- (0,22);
\draw [ultra thick] (-2,24.5) -- (-6,24);
\end{tikzpicture}}
\caption{The visualization of the Hardy example through its bundle diagram. This model has global events, but some local events can’t be extended to a global one.}
\label{Hardy Fig}
\end{figure}

\begin{table}
  \centering
\begin{tabular}{|c|c|c|c|c|}
\hline\noalign{\smallskip}
& \large{$00$} & \large{$01$} & \large{$10$} & \large{$11$} \\
\noalign{\smallskip}\hline\noalign{\smallskip}
\large{$ab$} & \large{$\frac{2}{9}$} & \large{$\frac{2}{3}$} & \large{$\frac{1}{9}$} & \large{$0$} \\
\noalign{\smallskip}\hline\noalign{\smallskip}
\large{$bc$} & \large{$0$} & \large{$\frac{1}{3}$} & \large{$\frac{2}{3}$} & \large{$0$} \\
\noalign{\smallskip}\hline\noalign{\smallskip}
\large{$cd$} & \large{$\frac{1}{3}$} & \large{$\frac{1}{3}$} & \large{$\frac{1}{3}$} & \large{$0$} \\
\noalign{\smallskip}\hline\noalign{\smallskip}
\large{$de$} & \large{$0$} & \large{$\frac{2}{3}$} & \large{$\frac{1}{3}$} & \large{$0$} \\
\noalign{\smallskip}\hline\noalign{\smallskip}
\large{$ea$} & \large{$\frac{2}{9}$} & \large{$\frac{1}{9}$} & \large{$\frac{2}{3}$} & \large{$0$} \\
\noalign{\smallskip}\hline
\end{tabular}
    \caption{Table of outcome probabilities of each context of the Hardy example.}
    \label{Hardy}
\end{table}
\end{example}

\begin{example}[Bell]
The model of Table \ref{Bell} is the usual example of a non-disturbing contextual model in quantum mechanics. It is also known as the Bell-CHSH model, by its application in the Clauser-Horne-Shimony-Holt inequality and for being an important example of bipartite scenario that shows non-local behavior \cite{bell_aspect_2004,PhysRevLett.23.880}. The $ab$ local section is trivial, but the others have probabilities that remember the liar cycle. In this sense, this model can be understood as a combination of the trivial example and three liars. Such decomposition agrees with the noncontextual fraction $NCF=\frac{3}{4}$, therefore the Bell example is a contextual model. As one can see by its bundle diagram, all events are defined as a restriction of global events. Therefore, it is an example of probabilistic contextuality \cite{Abramsky_2011}, once contextual behaviour only appears when considering the measures.

\begin{figure}
\centering
\scalebox{0.3}{
\begin{tikzpicture}
\draw [ultra thick] (-9,12) -- (-3,10) node[right] {\Huge{$\ a$}};
\draw [ultra thick] (-3,10) -- (0,13) node[right] {\Huge{$\ b$}}; 
\draw [ultra thick] (0,13) -- (-6,15)
node[left] {\Huge{$c\ $}};
\draw [ultra thick] (-6,15) -- (-9,12) node[left] {\Huge{$d\ $}};
\filldraw [black] (-9,12) circle (4pt);
\filldraw [black] (-3,10) circle (4pt);
\filldraw [black] (0,13) circle (4pt);
\filldraw [black] (-6,15) circle (4pt);
\draw[loosely dotted, ultra thick] (-9,12) -- (-9,21);
\draw[loosely dotted, ultra thick] (-3,10) -- (-3,19);
\draw[loosely dotted, ultra thick] (0,13) -- (0,22);
\draw[loosely dotted, ultra thick] (-6,15) -- (-6,24);
\filldraw [black] (-9,21) circle (4pt);
\filldraw [black] (-3,19) circle (4pt);
\filldraw [black] (0,22) circle (4pt);
\filldraw [black] (-6,24) circle (4pt);
\filldraw [black] (-9,19) circle (4pt);
\filldraw [black] (-3,17) circle (4pt);
\filldraw [black] (0,20) circle (4pt);
\filldraw [black] (-6,22) circle (4pt);
\draw [ultra thick] (-9,19) -- (-3,19);
\draw [ultra thick] (-3,17) -- (0,20) node[right] {\Huge{$\ 1$}}; 
\draw [ultra thick] (0,20) -- (-6,22); 
\draw [ultra thick] (-6,22) -- (-9,21);
\draw [ultra thick] (-9,21) -- (-3,17);
\draw [ultra thick] (-3,19) -- (0,22) node[right] {\Huge{$\ 0$}}; 
\draw [ultra thick] (0,22) -- (-6,24); 
\draw [ultra thick] (-6,24) -- (-9,19);
\draw [ultra thick] (-9,21) -- (-3,19);
\draw [ultra thick] (-9,19) -- (-3,17);
\draw [ultra thick] (-6,22) -- (-9,19);
\draw [ultra thick] (-6,24) -- (-9,21);
\draw [ultra thick] (0,22) -- (-6,22); 
\draw [ultra thick] (0,20) -- (-6,24); 
\end{tikzpicture}}
\caption{The visualization of the Bell example through its bundle diagram. This model has only global events, and the contextual behaviour only appears because of the impossibility to explain it with positive real numbers.}
    \label{Bell Fig}
\end{figure}

\begin{table}
  \centering
\begin{tabular}{|c|c|c|c|c|}
\hline\noalign{\smallskip}
& \large{$00$} & \large{$01$} & \large{$10$} & \large{$11$} \\
\noalign{\smallskip}\hline\noalign{\smallskip}
\large{$ab$} & \large{$\frac{1}{2}$} & \large{$0$} & \large{$0$} & \large{$\frac{1}{2}$} \\
\noalign{\smallskip}\hline\noalign{\smallskip}
\large{$bc$} & \large{$\frac{3}{8}$} & \large{$\frac{1}{8}$} & \large{$\frac{1}{8}$} & \large{$\frac{3}{8}$} \\
\noalign{\smallskip}\hline\noalign{\smallskip}
\large{$cd$} & \large{$\frac{3}{8}$} & \large{$\frac{1}{8}$} & \large{$\frac{1}{8}$} & \large{$\frac{3}{8}$} \\
\noalign{\smallskip}\hline\noalign{\smallskip}
\large{$da$} & \large{$\frac{1}{8}$} & \large{$\frac{3}{8}$} & \large{$\frac{3}{8}$} & \large{$\frac{1}{8}$} \\
\noalign{\smallskip}\hline
\end{tabular}
    \caption{Table of outcome probabilities of each context of the Bell example.}
    \label{Bell}
\end{table}
\end{example}

\section{Topology of measure bundles}
\label{III}

\subsection{Necessary condition to contextual behavior}
\label{IIIA}

A result discovered before the notion of contextuality being defined \cite{Vorobyev_1962}, known as Vorob'ev's theorem, is the characterization of a necessary but not sufficient condition in the scenario to be part of a contextual empirical model. It is linked with the notion of cyclicity, in the sense of Graham’s decomposition.

\begin{definition}
A hypergraph $\left(X,\mathcal{C}\right)$ is acyclic if it can be reduced to the empty set through the Graham’s decomposition, the algorithm given by the repeated application of the following operations:
\begin{itemize}
\item if $m\in X$ such that it belongs to a single hyperedge, then delete $m$ from this hyperedge;
\item if $V \subsetneq U$, with $V,U \in \mathcal{C}$, then delete $V$ from $\mathcal{C}$.
\end{itemize}
\end{definition}

The Graham’s decomposition algorithm can be interpreted as a coarse-graining of contexts, “forgetting” contexts that can be described by marginalization. Contextuality depends on maximal contexts and their intersections, as one can see by the contextual fraction algorithm. Graham’s decomposition deletes measurements that are not intersections and contexts that are not maximal, with the detail of also erasing the minimal context of each measurement from the hypergraph. In this sense, acyclic scenarios allow themselves to be simplified arbitrarily, while cyclic scenarios do not. An important fact is that Graham’s decomposition does not preserve the simplicial complex structure of the scenario as showed in Ref. \cite{barbosa2015contextuality}. The property of being cyclic, although it appears to be a topological property in the sense of capturing “holes” in the hypergraph, is not a topological invariant, as one can see in the next example.

\begin{example}
\label{ex9}
The first triangle of the Graham process in Fig. \ref{Vorobev1.png} is simply connected when equipped with the topology of its geometrical representation and therefore collapsible, as well as being acyclic when thought of as an hypergraph as shown. 
\begin{figure}[!ht]
\centering
\scalebox{0.6}{
\begin{tikzpicture}
\fill[gray] (-8,3) -- (-6,3) -- (-7,4.64) -- (-8,3);
\draw [ultra thick] (-8,3) -- (-6,3) -- (-7,4.64) -- (-8,3);
\filldraw [black] (-8,3) circle (4pt);
\filldraw [black] (-6,3) circle (4pt);
\filldraw [black] (-7,4.64) circle (4pt);
\draw [->, ultra thick, snake=snake, line after snake=1mm, segment amplitude=2mm] (-5,4) -- (-4,4);
\fill[gray] (-3,3) -- (-1,3) -- (-2,4.64) -- (-3,3);
\filldraw [black] (-3,3) circle (4pt);
\filldraw [black] (-1,3) circle (4pt);
\filldraw [black] (-2,4.64) circle (4pt);
\draw [->, ultra thick, snake=snake, line after snake=1mm, segment amplitude=2mm] (0,4) -- (1,4);
\node at (2,4) {\Huge{$\emptyset$}};
\end{tikzpicture}}
    \caption{Graham decomposition of a filled triangle. The steps are: to exclude the edges; to exclude the isolated vertices. The result is the empty set, therefore the hypergraph is acyclic.}
    \label{Vorobev1.png}
\end{figure}
In the case of the triangle with barycentric subdivision, the first of the Graham process in Fig. \ref{Vorobev0.png}, despite being simply connected when equipped with the topology of its geometrical representation and therefore collapsible, being homeomorphic to the previous case, it is not an acyclic hypergraph and can be a measurable scenario for an empirical model as we will see in a following example. 
\begin{figure}[!ht]
\centering
\scalebox{0.6}{
\begin{tikzpicture}
\fill[gray] (-8,3) -- (-6,3) -- (-7,4.64) -- (-8,3);
\draw [ultra thick] (-8,3) -- (-6,3) -- (-7,4.64) -- (-8,3);
\filldraw [black] (-8,3) circle (4pt);
\filldraw [black] (-6,3) circle (4pt);
\filldraw [black] (-7,4.64) circle (4pt);
\draw [ultra thick] (-7,3.6) -- (-8,3);
\draw [ultra thick] (-7,3.6) -- (-6,3);
\draw [ultra thick] (-7,3.6) -- (-7,4.64);
\filldraw [black] (-7,3.6) circle (4pt);
\draw [->, ultra thick, snake=snake, line after snake=1mm, segment amplitude=2mm] (-5,4) -- (-4,4);
\fill[gray] (-3,3) -- (-1,3) -- (-2,4.64) -- (-3,3);
\draw [ultra thick, white] (-2,3.6) -- (-3,3);
\draw [ultra thick, white] (-2,3.6) -- (-1,3);
\draw [ultra thick, white] (-2,3.6) -- (-2,4.64);
\filldraw [black] (-3,3) circle (4pt);
\filldraw [black] (-1,3) circle (4pt);
\filldraw [black] (-2,4.64) circle (4pt);
\filldraw [black] (-2,3.6) circle (4pt);
\draw [->, ultra thick, snake=snake, line after snake=1mm, segment amplitude=2mm] (0,4) -- (1,4);
\node at (2,4) {\Huge{\textbf{?}}};
\end{tikzpicture}}
\caption{Graham decomposition of a triangle with barycentric subdivision. The first and unique allowed step is to exclude the edges. There isn’t any other possible step, therefore the hypergraph is cyclic.}
\label{Vorobev0.png}
\end{figure}
\end{example}

This is a marked difference from the fiber bundle theory, where having a topologically trivial base (when equipped with a group structure) implies a topologically trivial total space \cite{Schwarz1994}. It also shows that the role of topology in contextuality, if any, should be more subtle.

We can rewrite Vorob’ev theorem as follows.

\begin{theorem}[Vorob’ev theorem]
Given base $\mathbf{X}$, any non-disturbing measure bundle defined on it is noncontextual if and only if its base seen as a hypergraph is acyclic.
\end{theorem} 

In short, contextuality does not follow directly from the standard topology of the scenario\footnote{But it can be related to topological aspects of the empirical model \cite{Okay_2020}.}, even with aciclicity structurally linked with noncontextuality \cite{Karvonen_2019}. The example in the Fig. \ref{Vorobev1.png} and Fig. \ref{Vorobev0.png} not only show independence in terms of topology, by detaching the topology of the simplicial complex from its ability to be a base for a contextual measure bundle, it also indicates that including freedom (the vertex in the center of the barycentric triangle) suffices to allow such behavior.

\subsection{n-contextuality}

\label{IIIB}

Given a measure bundle and once identified the simplicial complex structure of scenarios, one can ask the relation between the inductive definition of a simplicial complex, i.e. by increasing the dimension of the simplices to be included in each step, and the contextuality shown by the bundle for each step of such induction of the base. 

\begin{definition}
    Given a empirical model, its $0$-skeleton is given by the set of measurements $X$ and their sections; its $n$-skeleton for $n\geq 1$ is given by the set of measurements $X$ and the contexts up to $n+1$ measurements, and their sections.
\end{definition}

One can also define such construction as the marginalization of the empirical model on the set of contexts in the $n$-skeleton. The following definition addresses the question of contextuality induced by the $n$-skeleton by using contextual fraction.

\begin{definition}
Let $CF_{n}$ be the contextual fraction of the measure bundle induced by the marginalization of a given empirical model on the $n$-skeleton of its base. The $n$-contextuality is the difference $CF_{n}-CF_{n-1}$ between two steps of the inductive definition of the scenario.
\end{definition}

The concept of $n$-contextuality is well-defined once for vertices the scenario is acyclic, and therefore noncontextual, $CF_{0}=0$, and also because of the monotony with relation to free operations of a resource theory of contextuality \cite{amaral2019resource}. In special, the sum of the contextual fractions of disjoint sub-models is less than the contextual fraction of the entire model: any empirical model given by two disjoint parts $\mathcal{M}=\mathcal{M}_1\sqcup\mathcal{M}_2$ satisfies
\begin{equation}
CF(\mathcal{M})\geq CF(\mathcal{M}_1)+CF(\mathcal{M}_2)
\end{equation}
One can see that it generalizes to any finite disjoint union of sub-models. By construction of the contextual fraction, the following holds
\begin{equation}
    1\geq CF(\mathcal{M})\geq\sum_{k=1}^n CF(\mathcal{M}_k)
\end{equation}
for any empirical model $\mathcal{M}=\bigsqcup_{k=1}^n \mathcal{M}_k$. Therefore, adding edges can make contextuality appears, and we can define the sequence
\begin{equation}
    0=CF_0\leq CF_1\leq ...\leq CF_n=CF.
\end{equation}

$n$-contextuality is a way to explore the influence of topological properties of the scenario on the contextual behavior of the measure bundle, once one isolate by the dimension of the $n$-skeleton any topological property. A question in the literature \cite{Terra_2019_FibreBundle} between the topology of scenarios and acyclicity is whether all contextual behavior appears because of topological failures of the scenario captured by the first homology group, equivalently if there is $n$-contextuality for $n>1$, once the $1$-skeleton captures any property of the first homology group of the scenario. Previously we saw an example that split the notion of standard topology and cyclicity, but the question remains: for $n>1$, can contextuality appears purely due to $n$-dimensional objects? The next example shows that it does.

\begin{example}
\label{Tetra}
Let a scenario being given by the border of a tetrahedron is an example of a scenario that has non-trivial cyclicality, with only the four vertices, six edges, and four faces. Table \ref{tab: tetrahedron} has the probabilities of obtaining certain outcomes within maximal contexts, where we have the encoding of a bundle with identical fibers with two outputs $\left\{0,1\right\}$, and base with contexts represented by the boundary of a tetrahedron, with measurements $\left\{a,b,c,d\right\}$. It was construct from the Svetlichny’s box \cite{Barrett_2005} as presented in Ref. \cite{Soares_Barbosa_2014}, and described by the Table \ref{tab: Svetlichny} \footnote{This example was built independently by the author for this paper, but appeared before the publication of it in Ref. \cite{Dzhafarov_2020}.}
\begin{table*}\centering
\caption{Table of outcome probabilities of each context for the tetrahedron model.}
\label{tab: tetrahedron}
\begin{tabular}{|c|c|c|c|c|c|c|c|c|}
\hline\noalign{\smallskip}
& \large{$000$} & \large{$001$} & \large{$010$} & \large{$011$} & \large{$100$} & \large{$101$} & \large{$110$} & \large{$111$}\\
\noalign{\smallskip}\hline\noalign{\smallskip}
\large{$abc$} & \large{$\frac{1}{4}$} & \large{$0$} & \large{$0$} & \large{$\frac{1}{4}$} & \large{$0$} & \large{$\frac{1}{4}$} & \large{$\frac{1}{4}$} & \large{$0$}\\
\noalign{\smallskip}\hline\noalign{\smallskip}
\large{$abd$} & \large{$\frac{1}{4}$} & \large{$0$} & \large{$0$} & \large{$\frac{1}{4}$} & \large{$0$} & \large{$\frac{1}{4}$} & \large{$\frac{1}{4}$} & \large{$0$}\\
\noalign{\smallskip}\hline\noalign{\smallskip}
\large{$acd$} & \large{$0$} & \large{$\frac{1}{4}$} & \large{$\frac{1}{4}$} & \large{$0$} & \large{$\frac{1}{4}$} & \large{$0$} & \large{$0$} & \large{$\frac{1}{4}$}\\
\noalign{\smallskip}\hline\noalign{\smallskip}
\large{$bcd$} & \large{$0$} & \large{$\frac{1}{4}$} & \large{$\frac{1}{4}$} & \large{$0$} & \large{$\frac{1}{4}$} & \large{$0$} & \large{$0$} & \large{$\frac{1}{4}$}\\
\noalign{\smallskip}\hline
\end{tabular}
\end{table*}
As one can see, the probabilities in any subcontext represented by the edges are maximally random, implying noncontextuality when restricted to the $1$-skeleton. There is concordance between the faces, thus it is a non-disturbing empirical model. It is also a contextual empirical model once it has a non-null contextual fraction \cite{Abramsky_2017}, $CF=\frac{3}{4}$. In conclusion, the contextual behavior of this model does not originate from the $1$-skeleton but from the $2$-skeleton, thus this is an example of $2$-contextuality and a counter-example of the affirmation that contextuality follows, even indirectly, from the first holonomy group of the scenario with the standard topology of its geometrical representation.
\end{example}

Another question one can ask is whether quantum theory has any examples of $n$-contextuality for $n>1$. The next example shows quantum theory has a well-known example for all $n\geq 1$.

\begin{example}
The $n$ dimensional GHZ model is a scenario with $n+1$ parts, each with two measurements and fibers with two elements. For $n=1$ we have the Bell scenario, a graph with square form and an example of $1$-contextuality. For $n=2$ the scenario has an octahedron shape. The model is fixed when choosing the quantum state
\begin{equation}
    \ket{GHZ}=\frac{\ket{0}^{\otimes (n+1)}+\ket{1}^{\otimes (n+1)}}{\sqrt{2}}
\end{equation}
which is maximally entangled. By the Theorem \ref{FAB} and the structure of the measurements in quantum theory, it is always possible to find a set of measurements that has contextuality for this model. The marginalized measures are maximally uniform (a reflection of the state’s non-biseparability), and therefore the contextuality doesn’t appear in simplices of a smaller dimension than $n$. 
\end{example}

The case of the GHZ model in octahedron is well-known, and its usual version has the table of joint probabilities given in Table \ref{tab: GHZ}
\begin{table*}\centering
\caption{Table of outcome probabilities of each context for the GHZ model.}
\label{tab: GHZ}
\begin{tabular}{|c|c|c|c|c|c|c|c|c|}
\hline\noalign{\smallskip}
& \large{$000$} & \large{$001$} & \large{$010$} & \large{$011$} & \large{$100$} & \large{$101$} & \large{$110$} & \large{$111$}\\
\noalign{\smallskip}\hline\noalign{\smallskip}
\large{$ABC$} & \large{$0$} & \large{$\frac{1}{4}$} & \large{$\frac{1}{4}$} & \large{$0$} & \large{$\frac{1}{4}$} & \large{$0$} & \large{$0$} & \large{$\frac{1}{4}$}\\
\noalign{\smallskip}\hline\noalign{\smallskip}
\large{$ABc$} & \large{$\frac{1}{8}$} & \large{$\frac{1}{8}$} & \large{$\frac{1}{8}$} & \large{$\frac{1}{8}$} & \large{$\frac{1}{8}$} & \large{$\frac{1}{8}$} & \large{$\frac{1}{8}$} & \large{$\frac{1}{8}$}\\
\noalign{\smallskip}\hline\noalign{\smallskip}
\large{$AbC$} & \large{$\frac{1}{8}$} & \large{$\frac{1}{8}$} & \large{$\frac{1}{8}$} & \large{$\frac{1}{8}$} & \large{$\frac{1}{8}$} & \large{$\frac{1}{8}$} & \large{$\frac{1}{8}$} & \large{$\frac{1}{8}$}\\
\noalign{\smallskip}\hline\noalign{\smallskip}
\large{$Abc$} & \large{$\frac{1}{4}$} & \large{$0$} & \large{$0$} & \large{$\frac{1}{4}$} & \large{$0$} & \large{$\frac{1}{4}$} & \large{$\frac{1}{4}$} & \large{$0$}\\
\noalign{\smallskip}\hline\noalign{\smallskip}
\large{$aBC$} & \large{$\frac{1}{8}$} & \large{$\frac{1}{8}$} & \large{$\frac{1}{8}$} & \large{$\frac{1}{8}$} & \large{$\frac{1}{8}$} & \large{$\frac{1}{8}$} & \large{$\frac{1}{8}$} & \large{$\frac{1}{8}$}\\
\noalign{\smallskip}\hline\noalign{\smallskip}
\large{$aBc$} & \large{$\frac{1}{4}$} & \large{$0$} & \large{$0$} & \large{$\frac{1}{4}$} & \large{$0$} & \large{$\frac{1}{4}$} & \large{$\frac{1}{4}$} & \large{$0$}\\
\noalign{\smallskip}\hline\noalign{\smallskip}
\large{$abC$} & \large{$\frac{1}{4}$} & \large{$0$} & \large{$0$} & \large{$\frac{1}{4}$} & \large{$0$} & \large{$\frac{1}{4}$} & \large{$\frac{1}{4}$} & \large{$0$}\\
\noalign{\smallskip}\hline\noalign{\smallskip}
\large{$abc$} & \large{$\frac{1}{8}$} & \large{$\frac{1}{8}$} & \large{$\frac{1}{8}$} & \large{$\frac{1}{8}$} & \large{$\frac{1}{8}$} & \large{$\frac{1}{8}$} & \large{$\frac{1}{8}$} & \large{$\frac{1}{8}$}\\
\noalign{\smallskip}\hline
\end{tabular}
\end{table*}
which can be verified by contextual fraction that it is a strong contextual model ($NCF=0$). Its importance as a model has direct applications even in the foundations \cite{Lawrence_2023}, placing $n$-contextuality as an important phenomenon to validate interpretations of quantum theory. Another example also with strong contextuality, but super-quantum, that inspired the tetrahedron model in example \ref{Tetra} is the Svetlichny’s box of Table \ref{tab: Svetlichny}
\begin{table*}\centering
\caption{Table of outcome probabilities of each context for the Svetlichny’s box.}
\label{tab: Svetlichny}
\begin{tabular}{|c|c|c|c|c|c|c|c|c|}
\hline\noalign{\smallskip}
& \large{$000$} & \large{$001$} & \large{$010$} & \large{$011$} & \large{$100$} & \large{$101$} & \large{$110$} & \large{$111$}\\
\noalign{\smallskip}\hline\noalign{\smallskip}
\large{$ABC$} & \large{$\frac{1}{4}$} & \large{$0$} & \large{$0$} & \large{$\frac{1}{4}$} & \large{$0$} & \large{$\frac{1}{4}$} & \large{$\frac{1}{4}$} & \large{$0$}\\
\noalign{\smallskip}\hline\noalign{\smallskip}
\large{$ABc$} & \large{$\frac{1}{4}$} & \large{$0$} & \large{$0$} & \large{$\frac{1}{4}$} & \large{$0$} & \large{$\frac{1}{4}$} & \large{$\frac{1}{4}$} & \large{$0$}\\
\noalign{\smallskip}\hline\noalign{\smallskip}
\large{$AbC$} & \large{$\frac{1}{4}$} & \large{$0$} & \large{$0$} & \large{$\frac{1}{4}$} & \large{$0$} & \large{$\frac{1}{4}$} & \large{$\frac{1}{4}$} & \large{$0$}\\
\noalign{\smallskip}\hline\noalign{\smallskip}
\large{$Abc$} & \large{$0$} & \large{$\frac{1}{4}$} & \large{$\frac{1}{4}$} & \large{$0$} & \large{$\frac{1}{4}$} & \large{$0$} & \large{$0$} & \large{$\frac{1}{4}$}\\
\noalign{\smallskip}\hline\noalign{\smallskip}
\large{$aBC$} & \large{$\frac{1}{4}$} & \large{$0$} & \large{$0$} & \large{$\frac{1}{4}$} & \large{$0$} & \large{$\frac{1}{4}$} & \large{$\frac{1}{4}$} & \large{$0$}\\
\noalign{\smallskip}\hline\noalign{\smallskip}
\large{$aBc$} & \large{$0$} & \large{$\frac{1}{4}$} & \large{$\frac{1}{4}$} & \large{$0$} & \large{$\frac{1}{4}$} & \large{$0$} & \large{$0$} & \large{$\frac{1}{4}$}\\
\noalign{\smallskip}\hline\noalign{\smallskip}
\large{$abC$} & \large{$0$} & \large{$\frac{1}{4}$} & \large{$\frac{1}{4}$} & \large{$0$} & \large{$\frac{1}{4}$} & \large{$0$} & \large{$0$} & \large{$\frac{1}{4}$}\\
\noalign{\smallskip}\hline\noalign{\smallskip}
\large{$abc$} & \large{$0$} & \large{$\frac{1}{4}$} & \large{$\frac{1}{4}$} & \large{$0$} & \large{$\frac{1}{4}$} & \large{$0$} & \large{$0$} & \large{$\frac{1}{4}$}\\
\noalign{\smallskip}\hline
\end{tabular}
\end{table*}
which also has maximally uniform marginals, and therefore only has contextuality in dimension $n=2$.

The study of the algebraic behavior of $n$-contextuality will be explored in a future work.

\subsection{An example of topological influence}
\label{IIIC}

Armed with $n$-contextuality, we can explore whether the topology of the scenario interferes, even indirectly, with the amount of contextuality. Here, we will construct an example that allows us to observe this influence. Let’s start with an empirical model in a barycentric triangle, which ends up being a sub-model of the tetrahedron presented in example \ref{Tetra}. In it, the face $abc$ is forgotten and the noncontextual fraction is $NCF=\frac{1}{2}$, whereas the complete tetrahedron has $NCF=\frac{1}{4}$. We have two options: 
\begin{itemize}
    \item either the face adds $\frac{1}{4}$ of contextuality in the model,
    \item or it happens because of the change in topology.
\end{itemize} 

To test the possibilities, it is possible to study the local sections on the face, since the marginals are maximally uniform. Realizing that local sections in the triangle satisfy the system of linear equations
\begin{equation}
    \sum_{x}p(x,y,z|a,b,c)=p(y,z|b,c)
\end{equation}
\begin{equation}
    \sum_{y}p(x,y,z|a,b,c)=p(x,z|a,c)
\end{equation}
\begin{equation}
    \sum_{z}p(x,y,z|a,b,c)=p(x,y|a,b)
\end{equation}
\begin{equation}
    \sum_{x,y,z}p(x,y,z|a,b,c)=1
\end{equation}
that are valid due no-disturbance and the probability condition, we can see that the triangle $abc$ has a generic section on the face given by the vector
\begin{equation}
\left(\frac{1}{4}-\eta,\eta,\eta,\frac{1}{4}-\eta,\eta,\frac{1}{4}-\eta,\frac{1}{4}-\eta,\eta\right) 
\end{equation}
with $\eta\in[0,\frac{1}{4}]$ and the maximally uniform occurring in $\eta=\frac{1}{8}$. However, adding a face to the model one gets the generic case of Table \ref{tab: Tetra Generic}
\begin{table*}\centering
\caption{Table of outcome probabilities of each context for the tetrahedron model with generic context $abc$.}
\label{tab: Tetra Generic}
\begin{tabular}{|c|c|c|c|c|c|c|c|c|}
\hline\noalign{\smallskip}
& \large{$000$} & \large{$001$} & \large{$010$} & \large{$011$} & \large{$100$} & \large{$101$} & \large{$110$} & \large{$111$}\\
\noalign{\smallskip}\hline\noalign{\smallskip}
\large{$abc$} & \large{$\frac{1}{4}-\eta$} & \large{$\eta$} & \large{$\eta$} & \large{$\frac{1}{4}-\eta$} & \large{$\eta$} & \large{$\frac{1}{4}-\eta$} & \large{$\frac{1}{4}-\eta$} & \large{$\eta$}\\
\noalign{\smallskip}\hline\noalign{\smallskip}
\large{$abd$} & \large{$\frac{1}{4}$} & \large{$0$} & \large{$0$} & \large{$\frac{1}{4}$} & \large{$0$} & \large{$\frac{1}{4}$} & \large{$\frac{1}{4}$} & \large{$0$}\\
\noalign{\smallskip}\hline\noalign{\smallskip}
\large{$acd$} & \large{$0$} & \large{$\frac{1}{4}$} & \large{$\frac{1}{4}$} & \large{$0$} & \large{$\frac{1}{4}$} & \large{$0$} & \large{$0$} & \large{$\frac{1}{4}$}\\
\noalign{\smallskip}\hline\noalign{\smallskip}
\large{$bcd$} & \large{$0$} & \large{$\frac{1}{4}$} & \large{$\frac{1}{4}$} & \large{$0$} & \large{$\frac{1}{4}$} & \large{$0$} & \large{$0$} & \large{$\frac{1}{4}$}\\
\noalign{\smallskip}\hline
\end{tabular}
\end{table*}
where contextuality remains the same, with $NCF=\frac{1}{4}$, corroborating the option of topological dependency. Therefore, the simple addition of a face, regardless of which one, generates more contextuality, in line with the topological influence option. A formalization of this dependency must be constructed so that something more can be said about this result, which it’s beyond the scope of this paper.

\section{Exploring 1-contextual models}
\label{IV}

We will now investigate contextuality, restricting ourselves to contexts with two measurements in each maximal context. Such measurement scenarios can show only $1$-contextuality and can be represented as a graph with measurements as vertices and compatible pairs of measurements defining an edge. Our objective is to study the respective measure bundles and construct a way to identify contextuality analogous to differential geometry. We will introduce the concept of contextual holonomy, whose non-triviality under certain conditions is a sufficient condition for contextual behavior.

\subsection{Discrete differential geometry}
\label{IVA}

Even if the analogy between topological and measure bundles cannot be realized, one can use some ideas to study the latter. As a simplicial complex, any scenario accepts its study in the formalism of discrete differential geometry \cite{Crane:2013:DGP,grady2010discrete}. In this formalism, an $n$-simplex is treated as a ``quanta of space,'' and the topology follows from the topology of the simplicial complex in its geometric representation. Here we will present some useful concepts from this formalism\footnote{The presentation of discrete differential geometry here is far from formal and aims to instigate useful concepts.}.

We can define the boundary operator $\partial$ on an $n$-simplex in the usual way, through the orientation given to each $n$-simplex. Its boundary obtained by $\partial$ will be given by the sum of the $(n-1)$-simplices that form its boundary, that is, the $(n-1)$-simplices that can be obtained by marginalization. For example, in the case of $2$-simplices, the boundary operator allows us to consider the edge of a bi-dimensional simply connected region $S$, given by the sum of $1$-simplices
\begin{equation}
    \partial S=\sum_{i}a_{i}a_{i+1},
\end{equation}
where $a_{i}a_{i+1}$ are the $1$-simplices that make up an edge of $S$.

Another feature is the formalism of discrete differential forms, which can be informally conceived as functions to measure the “size” of the simplices. Let the set of $n$-simplices of a simplicial complex $\textbf{X}$ be denoted by $\mathcal{C}_{n}$. To define a “size” we must first fix the algebraic structure $A$ where it is valued, usually the real numbers or a Lie algebra. A discrete differential $n$-form is defined as a dual of $n$-simplices. We will denote by $\mathcal{C}^{n}$ the set of $n$-forms. If $\bra{\omega}\in\mathcal{C}^{n}$, then we have 
\begin{align}
\begin{split}
  \bra{\omega} \colon \mathcal{C}_{n} &\to A\\
  \ket{S} &\mapsto \left<\omega |U\right>.
 \end{split}
\end{align}

We can also define the coboundary operator that acts on the $n$-forms. To define how the coboundary operator acts, we need the generalized Stokes' theorem. For the case of $A$ being Abelian, we have that the coboundary operator $d$ is defined such that it satisfies the generalized Stokes' theorem
\begin{equation}
\bra{d\omega}\ket{S}=\bra{\omega}\ket{\partial S},
\end{equation}
while for the case of $A$ being non-Abelian, we define the coboundary operator $D$ through the non-Abelian version of the generalized Stokes' theorem
\begin{equation}
\bra{D\omega}\ket{S}=\bra{\omega}\ket{\partial S},
\end{equation} 
where we lose linearity.

With the $n$-forms, we can explore the pieces of the simplicial complex and extract information about its shape. An important result that holds in its discrete form is the Ambrose-Singer theorem, which relates the curvature of the simplicial complex and the parallel transport of the differential forms. In the discrete formalism, one can define discrete curvature by calculating the angle generated by the transport of the vector normal to the surface, which depends on the length of the path given by a differential form. This describes the geometry of the simplicial complex as immersed in a Euclidean space. Given that we use the $1$-form $\omega$ to find the ``size'' of the $1$-simplices of a closed path, by Stokes' theorem, we can represent the curvature of a surface $S$ as a $2$-form $d\omega:\mathcal{C}_{2}\to\mathbb{R}$ valued in real numbers measuring the phase generated by parallel transport. This phase defines an algebraic structure called holonomy, which encodes the phases of all closed paths from a fixed point, or even the phases of all parallel paths between two fixed points.

Measure bundles do not have a differential structure in their construction, and even the topology can be neglected, since contextuality may not be adequately expressed for the bundle topology to capture it. To see this, first remember example \ref{ex9}, which shows that the standard topology of the base is not sufficient to determine contextual behavior, although there is indeed some influence, as detected in section \ref{IIIC}. Furthermore, by observing bundles like the one in example \ref{Bell Fig}, we see that the topology of outcomes is no longer sufficient to identify contextual behavior.

Nevertheless, the approach of contextuality presented here has many similarities with discrete differential geometry. It is defined by a fixed $R$-state $P$, and since the no-disturbance condition implies that the $R$-valued measures are well defined, one can define measure-valued $1$-form on the simplicial complex as the dual of the simplices. Here we identify the elements of the $\sigma$-algebra of the outcome space $\textbf{O}^{U}$ of the context $U$ as the $1$-forms

\begin{equation}
    \bra{\omega}\ket{U}_{P}=P[U](\omega)=\mu^{U}(\omega)\in R,
\end{equation}
for all $\omega\in\Sigma^{U}$\footnote{In this sense, we are working with measurement contextuality. The state is kept fixed, and the contextual behavior follows from the relation between events and measurements, here described as contexts. Event and measurement define the notion of effect in operational theories, thus measurement contextuality is the contextuality of effects when the state is fixed.}. In the following, the notation of this section will be used when necessary.

\subsection{Connection}
\label{IVB}

Since the contextuality of $1$-contextual models appears only through the edges in the simplicial complex, we can use the joint distributions on the edges to define a connection in the measure bundle. This connection would encode, for each vertex and edge at this vertex, how the distribution modifies when carried from the edge along the vertex.

For each edge, we can define two Markov kernels, one for each direction on the edge, through the local section on it and the marginalizations of the complement measurement. They can be understood as the result of a path operator.

\begin{definition}
A path operator is a function $T::(U,m)\mapsto T_{nm}$ with $U\in\mathcal{C}_{1}$, $m,n\in\mathcal{C}_{0}$ and $m,n\in U$, that maps a $1$-simplex and one of its vertices to a Markov kernel $T_{nm}::\mu^{m}\mapsto\mu^{n}$.
\end{definition}

The joint distribution does not choose a direction in the $1$-simplex. Thus, to define a path operator, an initial vertex must be chosen. It is analogous to the discrete version of choosing a point and a direction in a variety and getting the rule of changing some property in that direction. Note, however, that $T_{nm}T_{mn}$ is not necessarily the identity since we are considering kernels, and they may not form a group. This fact is what differentiates a measure bundle from a topological bundle. In the latter, the path operator defines a deterministic map of the elements of the fiber, while in the former, we have a map of the distributions defined over such elements.

\begin{definition}
A connection is defined as $\mathbb{T}::s_{m}\mapsto T\left(\left\{m,*\right\},s_{m}(m)\right)$ where
\begin{widetext}
\begin{align}
  T ::\left(\left\{m,n\right\},s_{m}(m)\right)=\left(\left\{m,n\right\},\left(m,\textbf{O}^{m},\mu^{m}\right)\right) &\mapsto \left(n,\textbf{O}^{n},T_{nm}\mu^{m}\right)=s_{mn}(n),
\end{align}
\end{widetext}
with $T_{nm}$ being the Markov kernel. $T$ is called a path transformation.
\end{definition}

The path transformation is equivalent to the path operator but induced to work at the level of sections. The connection defined here is the discrete analog of the usual connection in differential geometry. To each section of a vertex, it gives all the path transformations with it as the initial point.

We will call a contextual connection the case where the connection is constructed with the information of the joint measure of the $1$-simplex. Given a context $\left\{m,n\right\}$ with two minimal contexts $m$ and $n$, the allowed path transformations have the form 
\begin{align}
  T::\left(\left\{m,n\right\},s_{m}(m)\right) &\mapsto s_{mn}(n),
\end{align}
with $s_{m}(m)=\left(m,\textbf{O}^{m},\mu^{m}\right)$ and $s_{mn}(n)=s_{mn}(mn)|_{n}=\left(n,\textbf{O}^{n},\mu^{mn}|_{n}\right)=\left(n,\textbf{O}^{n},T_{nm}\mu^{m}\right)$.

The connection defines a transformation in both the fibers and the base, but in such a way that we can separate it into the horizontal and vertical parts. The change in the fibers needs to be extracted from the Markov kernels, but the composition is not a globally well-defined product in the set of Markov kernels\footnote{They will form a category, known as the Markov category \cite{PANANGADEN1999171}.}. This will be a problem when seeking an algebraic structure that allows identifying the contextual behavior, as will become clear below.

To exemplify the contextual connection, let's generically construct the Markov kernels for dichotomic fibers with outcomes ${0,1}$ in the following example.

\begin{example}
Suppose a $1$-simplex with vertices $a$ and $b$, both with dichotomic fibers with outcomes ${0,1}$. As a context, there are the probabilities that each mutual result will occur, denoted by $p(i,j|a,b)$, with $i,j\in{0,1}$. We can construct the matrix of such probabilities $M^{(a,b)}_{ij}$, which has no notion of transport from one vertex to another. As we assume the action to the left of the Markov kernel of the path transformation, the origin vertex indexes the lines of the stochastic matrix. In addition, columns are normalized. In this way, the kernel will have forms dependent on the direction chosen in the $1$-simplex. We get by taking $a$ as the origin of the transport
\begin{align}
\begin{split}
    T_{ba}^{ij}&=\frac{M^{(a,b)}_{ij}}{\sum_{i}M^{(a,b)}_{ij}}\\
    &=\frac{p(i,j|a,b)}{\sum_{i}p(i,j|a,b)}\\
    &=\frac{p(i,j|a,b)}{p(j|b)}\\
    &=p(i|j;a,b)
\end{split}
\end{align}
and taking $b$ as the origin of the transport
\begin{align}
\begin{split}
T_{ab}^{ji}&=\frac{\left(M^{(a,b)}_{ij}\right)^{T}}{\sum_{i}M^{(a,b)}_{ij}}\\
&=\frac{p(i,j|a,b)}{\sum_{j}p(i,j|a,b)}\\
&=\frac{p(i,j|a,b)}{p(i|a)}\\
&=p(j|i;b,a).
\end{split}
\end{align}
\end{example}

\subsection{Parallel transport in 1-contextual models}
\label{IVC}

In possession of the connection in a model, we can look for ways to identify its $1$-contextuality. Inspired by the analogy between discrete differential geometry and how the connection acts on the vertices of the simplicial complex, we will explore the role of holonomy. We will start by defining parallel transport and the decomposition of such transport that we will use.

We will first study the case of a finite fiber $R$ with $n$ elements. Here, as described in Ref. \cite{amaral2017geometrical}, we can describe a probabilistic measure as a vector living in $\mathbb{R}^{d}$. In this space, we fix an orthonormal basis where each of its elements will represent an outcome, such that the conditions of probabilistic distributions (non-negative entries that sum to one) hold for the projections on this basis. We will keep the same fibers for the atomic measurements, thus the Markov kernels are stochastic matrices that represent the change in these vectors (a linear operator such that the probabilistic properties are preserved). With this description, we can identify the measure bundle as a vector bundle \cite{Lang1995}, with $\mathbb{R}^{d}$ as fibers. The $\sigma$-algebra of each fiber will be the discrete one, and the vectors accessible by a section are restricted to probabilistic ones.

First, let's study what kernels the no-disturbance condition allows in the vertices. The passage from one context to another can be described by a transition function that can be seen as a Markov kernel $T_{n_{f}m|_{m},n_{i}m|_{m}}$ such that
\begin{equation}
T_{n_{f}m|_{m},n_{i}m|_{m}}\left(\mu^{n_{i}m}|_{m}\right)=\left(\mu^{n_{f}m}|_{m}\right).
\end{equation}
Since $T_{n_{f}m|_{m},n_{i}m|_{m}}$ is a linear function that takes a vector representation of a measure to another, by the no-disturbance condition, it needs not just to preserve the distribution but also the basis of the vector space. Therefore, the Markov kernels that may have the general form $T_{mm}=ODO^{T}$, with $O$ being an orthogonal matrix and $D$ a diagonal matrix with entries of $0$s or $1$s, the definition of an idempotent matrix. Since we are going to assume non-singularity, and the only non-singular idempotent matrix is the identity, we can take the identity for all transitions and ignore them from now on.  

For the general case of a path transformation, the stochastic matrix can be singular value decomposed into an orthogonal part describing the transformations between frames and a diagonalizable matrix that encodes the information forgotten along the path. We can extract an element of $O(d)$ for each stochastic matrix, 
\begin{equation}
T_{mn}=U\Sigma V^{T}=\left(U\Sigma U^{T}\right)\left(UV^{T}\right),
\end{equation}
where the matrix $U$ is the eigenvector matrix with the orthonormal basis of $T_{mn}T^{T}_{mn}$, the matrix $V$ is the eigenvector matrix with the orthonormal basis of $T^{T}_{mn}T_{mn}$, and $\Sigma$ is the diagonal matrix with the singular values as the diagonal entries. We identify $UV^{T}$, which effectively changes the basis of the vector space and defines a monoid homomorphism
\begin{align}
\begin{split}
    \Phi:\mathbb{M}(d) &\to O(d)\\
    M &\mapsto UV^{T}
\end{split}
\end{align}
of the stochastic matrices to the orthogonal group, both defined in $\mathbb{R}^{d}$. As one can see, there will be problems when handling singular stochastic matrices, as they erase part of the original information.

Parallel transport on the measure bundle is the composition of the stochastic matrices from an initial vertex through a path of $1$-simplices to a final one, by two different paths. In the end, a comparison of the information transported is made. The information here is the basis of the vector space, or equivalently the frame in the induced frame bundle.

\begin{theorem}[Parallel transport and contextuality: finite case]
If a measure bundle satisfying the no-disturbance condition and with non-singular stochastic matrices is noncontextual, then for any two paths $I^{b}_{a}$ and $I'^{b}_{a}$ with the same initial and final vertices, it holds that
\begin{equation}
\prod_{I_{a}^{b}}\Phi\left(T_{a_{i}a_{i+1}}\right)=\prod_{I'^{b}_{a}}\Phi\left(T_{a_{i}a_{i+1}}\right).
\end{equation}
\label{parallelfinite}
\end{theorem}

\begin{proof}
The no-disturbance condition and noncontextuality imply the extension to a global section of the measures. Its measure bundle admits a monoid such that by the homomorphism $\Phi$ it can be reduced to a group $G=O(d)$. As all the outcome fibers are the same, we can write a local section on a measurement $a$ as being transported through a path $I_{a}^{b}$ as
\begin{equation}
\left(b,\textbf{O},\prod_{I_{a}^{b}}T_{a_{i}a_{i}}T_{a_{i}a_{i+1}}\mu^{a}\right),
\end{equation}
and through a path $I'^{b}_{a}$ as
\begin{equation}
\left(b,\textbf{O},\prod_{I'^{b}_{a}}T_{a_{i}a_{i}}T_{a_{i}a_{i+1}}\mu^{a}\right).
\end{equation}
By the no-disturbance condition, the phase of the parallel transport will be given by a change in the outcome space. By noncontextuality and Theorem \ref{FAB} we have a global section 
\begin{equation}
s_{X}(U)=\left(U,\textbf{O}^{X}|_{U},\mu^{X}|_{U}\right),
\end{equation}
such that they must be equal to the restriction of it to the vertex $b$, implying the equality
\begin{equation}
\prod_{I_{a}^{b}}\Phi\left(T_{a_{i}a_{i}}T_{a_{i}a_{i+1}}\right)=\prod_{I'^{b}_{a}}\Phi\left(T_{a_{i}a_{i}}T_{a_{i}a_{i+1}}\right).
\end{equation}
\end{proof}

For a non-finite outcome space, the representation in a vector space is more subtle, but we can do a similar construction. Suppose that the Markov kernels $T_{mn}=\omega(x,y)$ are square integrable on all the outcome fiber $O$. We can define
\begin{equation}
\omega_{1}(x,y)=\int_{R}\omega(x,z)\omega(y,z)dz
\end{equation}
and
\begin{equation}
\omega_{2}(x,y)=\int_{R}\omega(z,x)\omega(z,y)dz,
\end{equation}
which are symmetric and positive definite. Using the Mercer's theorem \cite{devito2011extension}, they can be rewritten as 
\begin{equation}
\omega_{1}(x,y)=\sum_{j}\alpha_{j}u_{j}(x)u_{j}(y)
\end{equation}
and
\begin{equation}
\omega_{2}(x,y)=\sum_{j}\alpha_{j}v_{j}(x)v_{j}(y).
\end{equation}
One can show that
\begin{equation}
\omega(x,y)=\sum_{j}\sqrt{\alpha_{j}}u_{j}(x)v_{j}(y).
\end{equation}
Therefore, the map $\Phi:\mathbb{M}\to\mathbb{M}$ can be defined as $\Phi\left(\omega(x,y)\right)=\sum_{j}u_{j}(x)v_{j}(y)$. With this, we can prove the next theorem similarly to the finite case.

\begin{theorem}[Parallel transport and contextuality: non-finite case]
If a measure bundle satisfying the no-disturbance condition, with square integrable and non-singular Markov kernels is noncontextual, then for any two paths $I_{a}^{b}$ and $I'^{b}_{a}$ with the same initial and final vertices, it holds that
\begin{equation}
\prod_{I_{a}^{b}}\Phi\left(T_{a_{i}a_{i+1}}\right)=\prod_{I'^{b}_{a}}\Phi\left(T_{a_{i}a_{i+1}}\right).
\end{equation}
\label{Parallel}
\end{theorem}

\subsection{Contextual holonomy}
\label{IVD}

Non-trivial parallel transport defines the holonomy group with respect to a fixed point. The holonomy group is non-trivial if and only if there is at least one non-trivial phase of the parallel transport through a loop. Given a loop $\gamma$ with a fixed point $x$ in the simplicial complex, we define the operator $P::\gamma\mapsto\prod_{\gamma}\Phi\left(T_{a_{i}a_{i+1}}\right)$.

\begin{definition}
The group $Hol_{x}(T)=\left\{P(\gamma)\in GL\left(d\right)|\gamma\text{ is a loop based at }x\right\}$ is called the holonomy group at the vertex $x\in\mathcal{C}^{-}$ of the connection $T$.
\end{definition}

The holonomy group will codify the phase of the non-trivial parallel transport as a subgroup generated in a fixed loop. Notice that $Hol_{x}(T)\in O(d)$, and that there is independence of the vertex $x$ in the loop, since a similarity transformation is enough to change the point without changing the group. As a result, we can then get the following.

\begin{theorem}
If an empirical model is $1$-noncontextual and does not present a singular path, then it has a trivial holonomy group for the contextual connection.
\label{Holonomy}
\end{theorem}

This follows from theorems \ref{parallelfinite} and \ref{Parallel} for closed loops, where the transition function given by the condition of no-disturbance and non-singularity imposes that the group must be trivial. Again, models with singular paths must be omitted, as they can erase contextual behavior in the holonomy by losing information. The reciprocal of this theorem is false: it is possible to get contextual models with a trivial holonomy group. Both cases will appear in the following examples.

We can express theorem \ref{Holonomy} through the logical statement $NC\wedge NS\implies TH$, where NC is noncontextuality, NS is non-singularity, and TH is trivial holonomy. The negation, $H\implies C\lor S$, says that holonomy as defined shows either contextuality or singularity in the model, which erases contextuality. The holonomy used here via singular value decomposition looks too coarse once it fixes the elements of the outgoing fiber as the base of the vector space. For this reason, it does not fully capture contextuality in non-singular models, besides the singularity having a strange behavior regarding contextuality\footnote{This is a problem when using group structures to study semiring-valued measures. Something is forgotten, and a different framework is necessary, see Ref. \cite{montanhano2021contextuality}.}. \footnote{One can speculate if higher holonomy groups can do an equivalent condition for higher contextual behaviors. The difficulties in dealing with higher holonomy groups and the necessity to handle probabilistic distributions seem to indicate that a different encoding of the measure bundle is necessary.}

\subsection{Examples}
\label{IVE}

\begin{example}[Trivial]
The example of Table \ref{Trivial bundle} has kernels
\begin{equation}T_{ij}=\begin{bmatrix}
1 & 0\\
0 & 1
\end{bmatrix},\end{equation}
implying $\Phi\left(T_{ij}\right)=id$, therefore
\begin{equation}\prod_{I_{x}}\Phi\left(T_{ij}\right)=id.\end{equation}
By no-disturbance, $\Phi\left(t_{xx}\right)=id$. The holonomy group of this example is trivial, in agreement with the noncontextual behaviour of this model.
\end{example}

\begin{example}[Liar cycle]
For the example of Table \ref{Liar cycle}, following the construction of the stochastic matrices, we get
\begin{equation}
T_{ab}=T_{bc}=T_{cd}=T_{de}=\begin{bmatrix}
1 & 0\\
0 & 1
\end{bmatrix},\end{equation}
therefore $\Phi\left(T_{ij}\right)=id$ for these contexts, but for the last one
\begin{equation}T_{ea}=\begin{bmatrix}
0 & 1\\
1 & 0
\end{bmatrix}\end{equation}
and $\Phi\left(T_{ea}\right)=T_{ea}$. As $\Phi\left(t_{xx}\right)=id$ for non-disturbing models, we get
\begin{equation}
\Phi\left(\prod_{I_{x}}T_{ij}\right)=\begin{bmatrix}
0 & 1\\
1 & 0
\end{bmatrix}\neq id.\end{equation}
The holonomy group is non-trivial, $Hol_{x}\left(T\right)=\mathbb{Z}_{2}$, because $\Phi\left(\prod_{2nI_{x}}T_{ij}\right)=id$ and $\Phi\left(\prod_{(2n+1)I_{x}}T_{ij}\right)=\begin{bmatrix}
0 & 1\\
1 & 0
\end{bmatrix}$, with $n\in\mathbb{Z}$, in agreement with the contextual behaviour of the model.
\end{example}

\begin{example}[KCBS]
The example of Table \ref{KCBS} has stochastic matrices
\begin{equation}
T_{ij}=\begin{bmatrix}
0 & 1\\
1 & 0
\end{bmatrix}
\end{equation}
for all contexts, resulting in $\Phi\left(\prod_{I_{x}}T_{ij}\right)=\begin{bmatrix}
0 & 1\\
1 & 0
\end{bmatrix}\neq id$, in agreement with its contextual behaviour. For the same reason of the Liar cycle example, the holonomy group is $Hol_{p}\left(T\right)=\mathbb{Z}_{2}$.
\end{example}

\begin{example}[Hardy]
Previously examples were simple, in the sense that $\Phi$ changed nothing in the stochastic matrices. The example of Table \ref{Hardy} cannot be understood just as a topological bundle. The stochastic matrices are
\begin{align}
\begin{split}
T_{ab}&=\begin{bmatrix}
\frac{2}{8} & 1\\
\frac{6}{8} & 0
\end{bmatrix}, T_{bc}=\begin{bmatrix}
0 & 1\\
1 & 0
\end{bmatrix}, T_{cd}=\begin{bmatrix}
\frac{1}{2} & 1\\
\frac{1}{2} & 0
\end{bmatrix}, \\
&T_{de}=\begin{bmatrix}
0 & 1\\
1 & 0
\end{bmatrix}, T_{ea}=\begin{bmatrix}
\frac{2}{3} & 1\\
\frac{1}{3} & 0
\end{bmatrix},
\end{split}
\end{align}
and doing the singular value decomposition and applying $\Phi$, we get 
\begin{equation}
\Phi\left(T_{ea}T_{de}T_{cd}T_{bc}T_{ab}\right)= \begin{bmatrix}
\frac{4}{5} & \frac{3}{5}\\
\frac{3}{5} & -\frac{4}{5}
\end{bmatrix},
\end{equation}
different of the identity implied by no-disturbance. By the same reason of the Liar cycle, the holonomy group is $Hol_{p}\left(T\right)=\mathbb{Z}_{2}$, a result of its contextual behaviour.
\end{example}

\begin{example}[Bell]
The stochastic matrices of Table \ref{Bell} are given by
\begin{align}
\begin{split}
T_{ab}&=\begin{bmatrix}
1 & 0\\
0 & 1
\end{bmatrix}, T_{bc}=T_{cd}=\begin{bmatrix}
\frac{3}{4} & \frac{1}{4}\\
\frac{1}{4} & \frac{3}{4}
\end{bmatrix}, \\
&T_{de}=\begin{bmatrix}
\frac{1}{4} & \frac{3}{4}\\
\frac{3}{4} & \frac{1}{4}
\end{bmatrix}, T_{ea}=\begin{bmatrix}
0 & 1\\
1 & 0
\end{bmatrix},
\end{split}
\end{align}
and $\Phi\left(T_{ea}T_{de}T_{cd}T_{bc}T_{ab}\right)=\begin{bmatrix}
0 & 1\\
1 & 0
\end{bmatrix}$. Therefore,  $Hol_{p}\left(T\right)=\mathbb{Z}_{2}$.
\end{example}

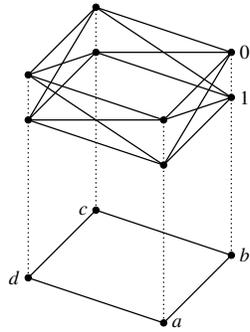
\begin{example}[Maximally random]
All the previous examples have non-singular stochastic matrices. This one, Table \ref{Homogeneous}, is the extreme case of singular paths.
\begin{figure}
  \centering
\scalebox{0.3}{
\begin{tikzpicture}

\draw [ultra thick] (-9,12) -- (-3,10) node[right] {\Huge{$\ a$}};
\draw [ultra thick] (-3,10) -- (0,13) node[right] {\Huge{$\ b$}}; 
\draw [ultra thick] (0,13) -- (-6,15)
node[left] {\Huge{$c\ $}}; 
\draw [ultra thick] (-6,15) -- (-9,12) node[left] {\Huge{$d\ $}};

\filldraw [black] (-9,12) circle (4pt);
\filldraw [black] (-3,10) circle (4pt);
\filldraw [black] (0,13) circle (4pt);
\filldraw [black] (-6,15) circle (4pt);

\draw[loosely dotted, ultra thick] (-9,12) -- (-9,21);
\draw[loosely dotted, ultra thick] (-3,10) -- (-3,19);
\draw[loosely dotted, ultra thick] (0,13) -- (0,22);
\draw[loosely dotted, ultra thick] (-6,15) -- (-6,24);

\filldraw [black] (-9,21) circle (4pt);
\filldraw [black] (-3,19) circle (4pt);
\filldraw [black] (0,22) circle (4pt);
\filldraw [black] (-6,24) circle (4pt);

\filldraw [black] (-9,19) circle (4pt);
\filldraw [black] (-3,17) circle (4pt);
\filldraw [black] (0,20) circle (4pt);
\filldraw [black] (-6,22) circle (4pt);

\draw [ultra thick] (-9,19) -- (-3,19);
\draw [ultra thick] (-3,19) -- (0,20) node[right] {\Huge{$\ 1$}}; 
\draw [ultra thick] (0,20) -- (-6,22); 
\draw [ultra thick] (-6,22) -- (-9,21);

\draw [ultra thick] (-9,21) -- (-3,17);
\draw [ultra thick] (-3,17) -- (0,22) node[right] {\Huge{$\ 0$}}; 
\draw [ultra thick] (0,22) -- (-6,24); 
\draw [ultra thick] (-6,24) -- (-9,19);

\draw [ultra thick] (-9,21) -- (-3,19);
\draw [ultra thick] (-9,19) -- (-3,17);

\draw [ultra thick] (-6,22) -- (-9,19);
\draw [ultra thick] (-6,24) -- (-9,21);

\draw [ultra thick] (0,22) -- (-6,22); 
\draw [ultra thick] (0,20) -- (-6,24); 

\draw [ultra thick] (-3,19) -- (0,22);
\draw [ultra thick] (-3,17) -- (0,20);

\end{tikzpicture}}
\caption{The visualization of the maximally random example through its bundle diagram. This model has all possible global sections of its scenario.}
    \label{Homogeneous Fig}
\end{figure}

\begin{table}
  \centering
\begin{tabular}{|c|c|c|c|c|}
\hline\noalign{\smallskip}
& \large{$00$} & \large{$01$} & \large{$10$} & \large{$11$} \\
\noalign{\smallskip}\hline\noalign{\smallskip}
\large{$ab$} & \large{$\frac{1}{4}$} & \large{$\frac{1}{4}$} & \large{$\frac{1}{4}$} & \large{$\frac{1}{4}$} \\
\noalign{\smallskip}\hline\noalign{\smallskip}
\large{$bc$} & \large{$\frac{1}{4}$} & \large{$\frac{1}{4}$} & \large{$\frac{1}{4}$} & \large{$\frac{1}{4}$} \\
\noalign{\smallskip}\hline\noalign{\smallskip}
\large{$cd$} & \large{$\frac{1}{4}$} & \large{$\frac{1}{4}$} & \large{$\frac{1}{4}$} & \large{$\frac{1}{4}$} \\
\noalign{\smallskip}\hline\noalign{\smallskip}
\large{$da$} & \large{$\frac{1}{4}$} & \large{$\frac{1}{4}$} & \large{$\frac{1}{4}$} & \large{$\frac{1}{4}$} \\
\noalign{\smallskip}\hline
\end{tabular}
    \caption{Table of outcome probabilities of each context of the maximally random example.}
    \label{Homogeneous}
\end{table}
Taking all the stochastic matrices as the same, 
\begin{equation}
T_{ij}=\begin{bmatrix}
\frac{1}{2} & \frac{1}{2}\\
\frac{1}{2} & \frac{1}{2}
\end{bmatrix},
\end{equation}
an idempotent matrix. We get $\Phi\left(T_{ij}\right)=id$, a noncontextual behaviour (which is confirmed by the contextual fraction), and a trivial holonomy group. 
\end{example}

\begin{example}[Modified Bell]
Bell example has a trivial path transformation, and in the modified example of Table \ref{Modified Bell} we change it to a "liar" path transformation. The point here is to show that a local change can force contextuality to vanish, getting a noncontextual non-disturbing model.
\begin{figure}
\centering
\scalebox{0.3}{
\begin{tikzpicture}

\draw [ultra thick] (-9,12) -- (-3,10) node[right] {\Huge{$\ a$}};
\draw [ultra thick] (-3,10) -- (0,13) node[right] {\Huge{$\ b$}}; 
\draw [ultra thick] (0,13) -- (-6,15)
node[left] {\Huge{$c\ $}}; 
\draw [ultra thick] (-6,15) -- (-9,12) node[left] {\Huge{$d\ $}};

\filldraw [black] (-9,12) circle (4pt);
\filldraw [black] (-3,10) circle (4pt);
\filldraw [black] (0,13) circle (4pt);
\filldraw [black] (-6,15) circle (4pt);

\draw[loosely dotted, ultra thick] (-9,12) -- (-9,21);
\draw[loosely dotted, ultra thick] (-3,10) -- (-3,19);
\draw[loosely dotted, ultra thick] (0,13) -- (0,22);
\draw[loosely dotted, ultra thick] (-6,15) -- (-6,24);

\filldraw [black] (-9,21) circle (4pt);
\filldraw [black] (-3,19) circle (4pt);
\filldraw [black] (0,22) circle (4pt);
\filldraw [black] (-6,24) circle (4pt);

\filldraw [black] (-9,19) circle (4pt);
\filldraw [black] (-3,17) circle (4pt);
\filldraw [black] (0,20) circle (4pt);
\filldraw [black] (-6,22) circle (4pt);

\draw [ultra thick] (-9,19) -- (-3,19);
\draw [ultra thick] (-3,19) -- (0,20) node[right] {\Huge{$\ 1$}}; 
\draw [ultra thick] (0,20) -- (-6,22); 
\draw [ultra thick] (-6,22) -- (-9,21);

\draw [ultra thick] (-9,21) -- (-3,17);
\draw [ultra thick] (-3,17) -- (0,22) node[right] {\Huge{$\ 0$}}; 
\draw [ultra thick] (0,22) -- (-6,24); 
\draw [ultra thick] (-6,24) -- (-9,19);

\draw [ultra thick] (-9,21) -- (-3,19);
\draw [ultra thick] (-9,19) -- (-3,17);

\draw [ultra thick] (-6,22) -- (-9,19);
\draw [ultra thick] (-6,24) -- (-9,21);

\draw [ultra thick] (0,22) -- (-6,22); 
\draw [ultra thick] (0,20) -- (-6,24); 

\end{tikzpicture}}
\caption{The visualization of the Modified Bell example through its bundle diagram. It's explicit the change of the local section of the context $ab$.}
    \label{Modified Bell Fig}
\end{figure}

\begin{table}
  \centering
\begin{tabular}{|c|c|c|c|c|}
\hline\noalign{\smallskip}
& \large{$00$} & \large{$01$} & \large{$10$} & \large{$11$} \\
\noalign{\smallskip}\hline\noalign{\smallskip}
\large{$ab$} & \large{$0$} & \large{$\frac{1}{2}$} & \large{$\frac{1}{2}$} & \large{$0$} \\
\noalign{\smallskip}\hline\noalign{\smallskip}
\large{$bc$} & \large{$\frac{3}{8}$} & \large{$\frac{1}{8}$} & \large{$\frac{1}{8}$} & \large{$\frac{3}{8}$} \\
\noalign{\smallskip}\hline\noalign{\smallskip}
\large{$cd$} & \large{$\frac{3}{8}$} & \large{$\frac{1}{8}$} & \large{$\frac{1}{8}$} & \large{$\frac{3}{8}$} \\
\noalign{\smallskip}\hline\noalign{\smallskip}
\large{$da$} & \large{$\frac{1}{8}$} & \large{$\frac{3}{8}$} & \large{$\frac{3}{8}$} & \large{$\frac{1}{8}$} \\
\noalign{\smallskip}\hline
\end{tabular}
    \caption{Table of outcome probabilities of each context of the Modified Bell example.}
    \label{Modified Bell}
\end{table}
We just change $T_{ab}$ of the Bell example to
\begin{equation}
T_{ab}=\begin{bmatrix}
0 & 1\\
1 & 0
\end{bmatrix},
\end{equation}
clearly changing the product to $T_{ea}T_{de}T_{cd}T_{bc}T_{ab}=id$, what is preserved by $\Phi$. The holonomy group is trivial, and the model is noncontextual, agreeing with the contextual fraction.
\end{example}

\begin{example}["Liar in maximally random" examples]
Another modification of a previous example. The example of Table \ref{Contextual Homogeneous} change a context of the homogeneous example, with the aim to show that even if the holonomy group is non-trivial, it is still possible for the model to be noncontextual because of singular paths.
\begin{figure}
\centering
\scalebox{0.3}{
\begin{tikzpicture}

\draw [ultra thick] (-9,12) -- (-3,10) node[right] {\Huge{$\ a$}};
\draw [ultra thick] (-3,10) -- (0,13) node[right] {\Huge{$\ b$}}; 
\draw [ultra thick] (0,13) -- (-6,15)
node[left] {\Huge{$c\ $}}; 
\draw [ultra thick] (-6,15) -- (-9,12) node[left] {\Huge{$d\ $}};

\filldraw [black] (-9,12) circle (4pt);
\filldraw [black] (-3,10) circle (4pt);
\filldraw [black] (0,13) circle (4pt);
\filldraw [black] (-6,15) circle (4pt);

\draw[loosely dotted, ultra thick] (-9,12) -- (-9,21);
\draw[loosely dotted, ultra thick] (-3,10) -- (-3,19);
\draw[loosely dotted, ultra thick] (0,13) -- (0,22);
\draw[loosely dotted, ultra thick] (-6,15) -- (-6,24);

\filldraw [black] (-9,21) circle (4pt);
\filldraw [black] (-3,19) circle (4pt);
\filldraw [black] (0,22) circle (4pt);
\filldraw [black] (-6,24) circle (4pt);

\filldraw [black] (-9,19) circle (4pt);
\filldraw [black] (-3,17) circle (4pt);
\filldraw [black] (0,20) circle (4pt);
\filldraw [black] (-6,22) circle (4pt);

\draw [ultra thick] (-9,19) -- (-3,19);
\draw [ultra thick] (-3,19) -- (0,20) node[right] {\Huge{$\ 1$}}; 
\draw [ultra thick] (0,20) -- (-6,22); 
\draw [ultra thick] (-6,22) -- (-9,21);

\draw [ultra thick] (-9,21) -- (-3,17);
\draw [ultra thick] (-3,17) -- (0,22) node[right] {\Huge{$\ 0$}}; 
\draw [ultra thick] (0,22) -- (-6,24); 
\draw [ultra thick] (-6,24) -- (-9,19);

\draw [ultra thick] (-9,21) -- (-3,19);
\draw [ultra thick] (-9,19) -- (-3,17);

\draw [ultra thick] (-6,22) -- (-9,19);
\draw [ultra thick] (-6,24) -- (-9,21);

\draw [ultra thick] (0,22) -- (-6,22); 
\draw [ultra thick] (0,20) -- (-6,24); 

\end{tikzpicture}}
\caption{The visualization of the "Liar in maximally random" example with three maximally random paths. Its bundle diagram is the same of the Modified Bell example, the difference is in the probabilistic weights.}
    \label{Contextual Homogeneous Fig}
\end{figure}

\begin{table}
  \centering
\begin{tabular}{|c|c|c|c|c|}
\hline\noalign{\smallskip}
& \large{$00$} & \large{$01$} & \large{$10$} & \large{$11$} \\
\noalign{\smallskip}\hline\noalign{\smallskip}
\large{$ab$} & \large{$0$} & \large{$\frac{1}{2}$} & \large{$\frac{1}{2}$} & \large{$0$} \\
\noalign{\smallskip}\hline\noalign{\smallskip}
\large{$bc$} & \large{$\frac{1}{4}$} & \large{$\frac{1}{4}$} & \large{$\frac{1}{4}$} & \large{$\frac{1}{4}$} \\
\noalign{\smallskip}\hline\noalign{\smallskip}
\large{$cd$} & \large{$\frac{1}{4}$} & \large{$\frac{1}{4}$} & \large{$\frac{1}{4}$} & \large{$\frac{1}{4}$} \\
\noalign{\smallskip}\hline\noalign{\smallskip}
\large{$da$} & \large{$\frac{1}{4}$} & \large{$\frac{1}{4}$} & \large{$\frac{1}{4}$} & \large{$\frac{1}{4}$} \\
\noalign{\smallskip}\hline
\end{tabular}
    \caption{Table of outcome probabilities of each context with three maximally random paths of the "Liar in maximally random" example.}
    \label{Contextual Homogeneous}
\end{table}
In the Maximally random example with three maximally random paths, we change $T_{ab}$ to
\begin{equation}
T_{ab}=\begin{bmatrix}
0 & 1\\
1 & 0
\end{bmatrix},
\end{equation}
which implies $\prod_{I_{x}}\Phi\left(T_{ij}\right)=T_{ab}$, a non-trivial parallel transport. Again and for the same previously reason, $Hol_{p}\left(T\right)=\mathbb{Z}_{2}$. But $NCF=1$, and it is a noncontextual model. We could exchange a maximally random path for identity one, not changing the holonomy group. For example, let’s change the path $bc$ to the trivial one, $T_{bc}=id$, keeping two maximally random paths. Though contextual fraction one can find that the model is still noncontextual. But if we change another maximally random path, say $cd$, to the trivial one, one gets $NCF=\frac{1}{2}$, implying contextual behavior.
\end{example}

\begin{example}[Counterexample for trivial holonomy]
An example that has no singularity, but has contextuality even though it does not have holonomy, Table \ref{Contextual trivial}.
\begin{figure}
  \centering
\scalebox{0.3}{
\begin{tikzpicture}

\draw [ultra thick] (-9,12) -- (-3,10) node[right] {\Huge{$\ a$}};
\draw [ultra thick] (-3,10) -- (0,13) node[right] {\Huge{$\ b$}}; 
\draw [ultra thick] (0,13) -- (-6,15)
node[left] {\Huge{$c\ $}}; 
\draw [ultra thick] (-6,15) -- (-9,12) node[left] {\Huge{$d\ $}};

\filldraw [black] (-9,12) circle (4pt);
\filldraw [black] (-3,10) circle (4pt);
\filldraw [black] (0,13) circle (4pt);
\filldraw [black] (-6,15) circle (4pt);

\draw[loosely dotted, ultra thick] (-9,12) -- (-9,21);
\draw[loosely dotted, ultra thick] (-3,10) -- (-3,19);
\draw[loosely dotted, ultra thick] (0,13) -- (0,22);
\draw[loosely dotted, ultra thick] (-6,15) -- (-6,24);

\filldraw [black] (-9,21) circle (4pt);
\filldraw [black] (-3,19) circle (4pt);
\filldraw [black] (0,22) circle (4pt);
\filldraw [black] (-6,24) circle (4pt);

\filldraw [black] (-9,19) circle (4pt);
\filldraw [black] (-3,17) circle (4pt);
\filldraw [black] (0,20) circle (4pt);
\filldraw [black] (-6,22) circle (4pt);

\draw [ultra thick] (-9,19) -- (-3,19);
\draw [ultra thick] (-3,19) -- (0,22) node[right] {\Huge{$\ 1$}}; 
\draw [ultra thick] (0,20) -- (-6,22); 

\draw [ultra thick] (-9,21) -- (-3,17);
\draw [ultra thick] (-3,17) -- (0,20) node[right] {\Huge{$\ 0$}}; 
\draw [ultra thick] (0,22) -- (-6,24); 

\draw [ultra thick] (-9,21) -- (-3,19);
\draw [ultra thick] (-9,19) -- (-3,17);

\draw [ultra thick] (-6,22) -- (-9,19);
\draw [ultra thick] (-6,24) -- (-9,21);


\end{tikzpicture}}
\caption{The visualization via bundle diagram of the counterexample for trivial holonomy. The idea is to keep all edges trivial, except one that has the same probabilities as the Bell model.}
    \label{Contextual trivial Fig}
\end{figure}

\begin{table}
  \centering
\begin{tabular}{|c|c|c|c|c|}
\hline\noalign{\smallskip}
& \large{$00$} & \large{$01$} & \large{$10$} & \large{$11$} \\
\noalign{\smallskip}\hline\noalign{\smallskip}
\large{$ab$} & \large{$\frac{1}{2}$} & \large{$0$} & \large{$0$} & \large{$\frac{1}{2}$} \\
\noalign{\smallskip}\hline\noalign{\smallskip}
\large{$bc$} & \large{$\frac{1}{2}$} & \large{$0$} & \large{$0$} & \large{$\frac{1}{2}$} \\
\noalign{\smallskip}\hline\noalign{\smallskip}
\large{$cd$} & \large{$\frac{1}{2}$} & \large{$0$} & \large{$0$} & \large{$\frac{1}{2}$} \\
\noalign{\smallskip}\hline\noalign{\smallskip}
\large{$da$} & \large{$\frac{3}{8}$} & \large{$\frac{1}{8}$} & \large{$\frac{1}{8}$} & \large{$\frac{3}{8}$} \\
\noalign{\smallskip}\hline
\end{tabular}
    \caption{Table of outcome probabilities of each context of the counterexample for trivial holonomy.}
    \label{Contextual trivial}
\end{table}
For three of the four edges holds $T_{ab}=T_{bc}=T_{cd}=Id$. The remain has the form
\begin{equation}
T_{ab}=\begin{bmatrix}
\frac{3}{4} & \frac{1}{4}\\
\frac{1}{4} & \frac{3}{4}
\end{bmatrix},
\end{equation}
which implies $\prod_{I_{x}}\Phi\left(T_{ij}\right)=Id$, a trivial parallel transport. The holonomy group is trivial, but $NCF=\frac{3}{4}$, and therefore it is a contextual model.
\end{example}

\subsection{Discrete exterior derivative in 1-contextuality}
\label{IVF}

The $\Phi$ image of a Markov kernel will be written as $\Phi\left(T_{nm}\right)=e^{\nabla_{nm}}$, with $\nabla_{nm}$ the average covariant derivative on the $1$-simplex $nm$. We then define the connection of the frame bundle as the $1$-form valued in the Lie algebra $\mathfrak{g}$
\begin{align}
\begin{split}
    \nabla:\mathcal{C}_{1} &\to\mathfrak{g}\\
    nm &\mapsto \nabla_{nm}.
\end{split}
\end{align}
An $1$-form does not act only on the $1$-simplex, but depends on its orientation, here implicit in the ordering $nm$.

Locally, in the sense of belonging to a context, the connection $1$-form can be write as its vertical and horizontal parts,
\begin{equation}
\nabla_{nm}=\triangle_{nm}+A_{nm}.
\end{equation}
with the horizontal part
\begin{align}
\begin{split}
    \triangle:\mathcal{C}_{1}&\to\mathfrak{g}\\
    nm &\mapsto \triangle_{nm}
\end{split}
\end{align}
and the vertical part
\begin{align}
\begin{split}
    A:\mathcal{C}_{1}&\to\mathfrak{g}\\
    nm &\mapsto A_{nm}
\end{split}
\end{align}
being $1$-forms as well. These two parts commute in the group representation, implying that their sum in the algebra representation is well defined.

The product of an $1$-form $\omega$ with a $1$-simplex $nm$ to a Lie algebra will be denoted by
\begin{equation}
    \bra{\omega}\ket{nm}=log\left(e^{\omega\left(nm\right)}\right).
\end{equation}
The boundary operator allows us to write the product over a loop, making sense of the notation $log(exp)$ to symbolize the Baker–Campbell–Hausdorff formula for a non-Abelian algebra,
\begin{equation}
    \bra{\omega}\ket{\partial S}=log\left(\prod_{i}e^{\omega\left(a_{i}a_{i+1}\right)}\right). 
\end{equation}
By the coboundary operator we can define the $2$-forms $DA$ and $dA$, respectively the non-Abelian and Abelian cases, as the discrete versions of curvature that codify the contextual behaviour on the frame bundle. With a curvature we can write from a boundary $\partial S$
\begin{align}
\begin{split}
\bra{\nabla}\ket{\partial S}&=\bra{\triangle}\ket{\partial S}+\bra{A}\ket{\partial S}\\
&=\bra{A}\ket{\partial S}\\
&=\bra{DA}\ket{S},
\end{split}
\end{align}
where the second equality follows from $\bra{\triangle}\ket{\partial S}=0$ once it depends only the base of the bundle, and $S$ is the connected surface defined by immersion the base in a Euclidean space. From Theorem \ref{Holonomy}, noncontextuality implies the discrete contextual curvature is trivial. But again, the reciprocal doesn't hold.

What we are doing here is to translate, even unfaithfully, the $1$-contextuality of a measure bundle into a topological property of a frame bundle, its holonomy, here embodied by its curvature. It is interesting to see the rule of curvature as the “generator of contextuality” in such representation. One could thus thing the measure bundle as a kind of “generalization” of the usual frame bundle by allowing not just orthogonal transformations but any stochastic transformations. The study of possible "contextuality generators" and their relationship with the bundle's topology will be the focus of future work.

\subsection{Disturbing empirical models}
\label{Disturbing empirical models}

For disturbing empirical models, the marginal measures could be different, thus the definition of noncontextuality does not hold. Here we will summarize the analogy between manifolds and empirical models that inspired the topological study of contextuality. Then we will present a conjecture that extends the relation of holonomy and $1$-contextual behavior to models with disturbance.

The analogy can be seen from the definitions shown in Table \ref{tabeladeanalogia}. The idea of this analogy follows from the fact that, in the finite case, a context can be seen as a Boolean structure. This is equivalent to having a discrete topology on the context measurements, with opens being subcontexts. The outcomes also have a discrete topology, with the events defining opens now in the topological bundle. Although the triviality of the topological bundle follows directly from the triviality of the measurable bundle, the same cannot be said for the measure bundle. This can only be defined as locally trivial by the contexts, which determine a cover of the measurement set. We have a structure analogous to a manifold, but not $\mathbb{R}^n$ as a reference for local isomorphism, but a classical description of the distributions imposed by the compatibility of the measurements. Noncontextuality becomes analogous to the question of the existence of an isomorphism of the manifold in $\mathbb{R}^n$. The transition functions are immediately identified as the objects that encode the disturbance in the model, as they describe how the model behaves in the intersections of the contexts.
\begin{table*}[]
    \centering
    \begin{tabular}{|p{0.45\textwidth}|p{0.45\textwidth}|}
\noalign{\smallskip}\hline\noalign{\smallskip}
\textbf{Differentiable manifolds} & \textbf{Empirical models} \\
\noalign{\smallskip}\hline\noalign{\smallskip}
        $M$ is a Hausdorff topological space with a countable basis & $M$ is a measure bundle \\
\noalign{\smallskip}\hline\noalign{\smallskip}
        $\mathcal{U}$ is a collection of homeomorphisms $x:U\to \mathbb{R}^m$ of open sets $U\subset M$ onto open sets $x(U)\subset \mathbb{R}^m$ & $\mathcal{C}$ is a collection of contexts of $M$ \\
\noalign{\smallskip}\hline\noalign{\smallskip}
        The domains $U$ of the homeomorphisms $x\in\mathcal{U}$ cover $M$ & The contexts $U$ in $C$ cover $M$ \\
\noalign{\smallskip}\hline\noalign{\smallskip}
        Given $x:U\to \mathbb{R}^m$ and $y:V\to \mathbb{R}^m$ of $\mathcal{U}$ with $U\cap V\neq\emptyset$, then $\phi_{xy}:x(U\cap V)\to y(U\cap V)$ is a homeomorphism of class $C^k$ & Given $s_U$ and $s_V$ with $U,V\in\mathcal{C}$ and $U\cap V\neq\emptyset$, then $t_{UV}:s_U(U\cap V)\to s_V(U\cap V)$ is an isomorphism \\
\noalign{\smallskip}\hline\noalign{\smallskip}
        Given a homeomorphism $z:W\to \mathbb{R}^m$ of an open set $W\subset M$ onto an open $z(W)\subset \mathbb{R}^m$, such that $\phi_{zx}$ and $\phi_{xz}$ are of class $C^k$ for each $x\in\mathcal{U}$, then $z\in\mathcal{U}$
        & Given $s_W$ with $W\in M$, such that $t_{WU}$ and $t_{UW}$ exist for all $s_U$, then $W\in\mathcal{C}$ \\
\noalign{\smallskip}\hline\noalign{\smallskip}
    \end{tabular}
    \caption{Table showing the analogy between the definitions of differentiable manifolds (left) and empirical models (right). For simplicity, we kept the column for empirical models in the finite case. Note that the notion of disturbance is analogous to non-trivial transition maps between contexts.}
    \label{tabeladeanalogia}
\end{table*}

\begin{definition}
For a disturbing empirical model, the transition function between contexts is defined by the application
\begin{align}
\begin{split}
t_{UV}:\Gamma\left(U\cap V\right)&\to\Gamma\left(U\cap V\right)\\
\left(U\cap V,\textbf{O}^{U\cap V},\mu^{V}|_{U\cap V}\right)&\mapsto\left(U\cap V,\textbf{O}^{U\cap V},\mu^{U}|_{U\cap V}\right).
\end{split}
\end{align}
\end{definition}

The transition function $t_{UV}$ is a Markov kernel when describing disturbance in models that can be codified in $1$-simplices, where the intersection is a vertex that links two paths. The disturbance can be understood as the non-trivial contribution of the vertex in the path transformation. Using the same argumentation as before, we conjecture the following.

\begin{conjecture}[Parallel transport and contextuality: disturbing case]
If a measure bundle with square integrable Markov kernels and with non-singular stochastic matrices is noncontextual, then for any two paths $I_{a}^{b}$ and $I'^{b}_{a}$ with the same initial and final vertices, the following holds:
\begin{equation}
\prod_{I_{a}^{b}}\Phi\left(t_{a_{i}a_{i}}T_{a_{i}a_{i+1}}\right)=
\prod_{I'^{b}_{a}}\Phi\left(t_{a_{i}a_{i}}T_{a_{i}a_{i+1}}\right).
\label{Disturbing}
\end{equation}
\end{conjecture}

The holonomy group for $1$-contextual models remains the same. The transition functions need to be taken into account when defining the group, and the description by holonomy remains valid. For a loop $\gamma$ based in $x\in\mathcal{C}_{0}$ of a scenario, the element of the holonomy group $Hol_x\left(T\right)$ for this loop has the form
\begin{equation}
    P\left(\gamma\right)=\prod_{\gamma}\Phi\left(t_{a_{i}a_{i}}T_{a_{i}a_{i+1}}\right).
\end{equation}
For a non-Abelian group, the product order matters, but for Abelian ones, we can write the elegant equation
\begin{equation}
    P\left(\gamma\right)=\prod_{\gamma}\Phi\left(t_{a_{i}a_{i}}\right)\prod_{\gamma}\Phi\left(T_{a_{i}a_{i+1}}\right).
\end{equation}
We can calculate directly from the table of probabilities for the local sections of contexts with two measurements, and from the transformation data by the disturbance of the model, the elements of the holonomy group.

The conjecture is loosely similar to the extended version of the empirical model in Ref. \cite{Amaral_2019} for disturbing models that can be codified in $1$-simplices. We can just write an extended empirical model where $t_{a_{i}a_{i}}$ are just another (virtual) context and work in a non-disturbing model.

\section{Discussion}
\label{V}

The bundle approach to contextuality is an important way to visualize the phenomenon of contextuality in a generic manner, and it allows a different path to explore such a phenomenon. Our first objective in this work was to construct such an approach for the case of non-finite fibers, indicating why scenarios are usually encoded in simplicial complexes, the different notions of noncontextuality, and their unification by the Fine–Abramsky–Brandenburger theorem.

Our second objective was to explore the relation of topology and contextual behavior, with the definition of $n$-contextuality. The example of the GHZ model shows that all levels of the hierarchy are presented in quantum theory, and the tetrahedron scenario gives hints of a relation between non-trivial topology and contextuality. We can perceive that although the usual topology of the base of the measure bundle does not have a direct relation with contextual behavior, it can still influence such behavior. A more suitable topology might be the key to bringing this influence to light.

Our last objective was to explore the first non-trivial level of $n$-contextuality. By restricting ourselves to the behavior generated by the edges, the analogy with discrete differential geometry allowed us to construct the notion of connection, and this led us to the relationship between contextual behavior and holonomy. Such a relation permits defining a curvature that appears in contextual models. The analogy with differential geometry becomes deeper when focusing on topological manifolds, which may be very useful in the future for considering models with disturbance.

One of the questions in Ref. \cite{Terra_2019_FibreBundle} is whether contextuality can be fully described through the first homology group or not. As presented here, specifically with the tetrahedron scenario and the GHZ models, there is contextuality that depends on higher homology groups, culminating in the notion of $n$-contextuality. More work must be done to construct identification tools to turn this concept into a clearer framework. The definition of $n$-contextuality raises the possibility of also studying the disturbance with its dependence on the dimension of the simplicial complexes, an “$n$-disturbance”. A deeper study of these hierarchies, as well as their formal definition, are open problems.

Tools to identify contextuality already exist, but the contextual connection has a straightforward interpretation and a visual appeal. There are many paths for exploring the relationship of this dynamic with $1$-contextuality. The role that singular paths, paths with Markov kernels that present singularity, have in erasing contextual behaviors should be better explored, enabling a concept that identifies contextuality under necessary and sufficient conditions. For this, other ways to extract holonomy will be necessary, especially one that is more general than an algebraic group. The known problem of using group structure, since it is not natural when handling probabilities, implies violations in the characterization of contextuality \cite{montanhano2021contextuality}.

To conclude, a generalization that exposes the role of higher holonomy groups in $n$-contextuality is a more challenging but more rewarding path. It would allow a general study of the hierarchy, its topological dependence, and the characterization of contextuality. For disturbing models, the codification of extended contextuality in transition functions and its role in the holonomy group deserves future attention. This path will likely require a categorical construction \cite{Karvonen_2019}, which could allow the study of non-trivial measurable bundles and the exact data one loses when studying the frame bundles and holonomy.

\begin{acknowledgments}
The author thanks Marcelo Terra Cunha, Rui Soares Barbosa and Hamadalli de Camargo for helpful comments and discussions. 

This study was financed in part by the Coordenação de Aperfeiçoamento de Pessoal de Nível Superior – Brasil (CAPES) – Finance Code 001.
\end{acknowledgments}

\bibliographystyle{quantum}
\bibliography{bib}

\end{document}